\documentclass[11pt,a4paper]{amsart}

\usepackage[a4paper,left=2.25cm,right=2.25cm,top=3cm,bottom=3cm]{geometry}
\usepackage{tikz}
\usepackage{hyperref}
\allowdisplaybreaks[1]
\theoremstyle{definition}
\newtheorem{thm}{Theorem}[section]
\newtheorem{prop}[thm]{Proposition}

\newtheorem{cor}[thm]{Corollary}
\newtheorem{dfn}[thm]{Definition}
\newtheorem{lem}[thm]{Lemma}

\makeatletter

\@addtoreset{equation}{section}
\makeatother

\def\Z{\mathbb{Z}}

\def\C{\mathbb{C}}

\newcommand{\bra}[1]{\left\langle #1\right|}
\newcommand{\ket}[1]{\left|#1\right\rangle}

\def\Wop[#1][#2][#3][#4][#5]{W\left(\left.\begin{array}{cc} #1&#2\\#3&#4\end{array}\right|#5\right)}
\def\Wopc[#1][#2][#3][#4][#5]{W^{3C}\left(\left.\begin{array}{cc} #1&#2\\#3&#4\end{array}\right|#5\right)}

\def\f[#1]{{\color{black} f}(#1)}
\def\fj[#1]{{\color{black} f}^{(j)}\left(#1\right)}

\newcommand{\bC}{\overline{C}}
\newcommand{\bI}{{\overline{I}}}
\newcommand{\bS}{{\overline{S}}}
\def\cB{\mathcal{B}}
\def\cH{\mathcal{H}}

\newcommand{\bPsi}{\overline{\Psi}}
\newcommand{\bpsi}{\overline{\psi}}
\newcommand{\balpha}{\overline{\alpha}}
\newcommand{\ba}{{\bf a}}
\newcommand{\bu}{{\bf u}}
\newcommand{\bv}{{\bf v}}
\newcommand{\bi}{{\bf i}}

\newcommand{\be}{\begin{eqnarray}}
\newcommand{\ee}{\end{eqnarray}}
\newcommand{\ben}{\begin{eqnarray*}}
\newcommand{\een}{\end{eqnarray*}}
\newcommand{\bec}{\begin{equation}\begin{array}{lll}}
\newcommand{\eec}{\end{array}\end{equation}}

\newcommand{\ot}{\otimes}
\newcommand{\sli}{\sum\limits}

\newcommand{\mref}[1]{(\ref{#1})}

\title{Theta function solutions of the qKZB equation for a face model}
\author{Peter E. Finch}
\author{Robert Weston}
\author{Paul Zinn-Justin}
\address{Peter E. Finch and Paul Zinn-Justin, CNRS UMR 7589 / Sorbonne Universit\'es, UPMC Univ Paris 06, LPTHE, F-75005, Paris, France \bigskip}
\address{Robert Weston, Department of Mathematics, Heriot-Watt University, Edinburgh EH14 4AS, UK,  and Maxwell Institute for Mathematical Sciences, Edinburgh, U.K.}

\thanks{P.~Finch and P.~Zinn-Justin are supported by ERC grant 278124 ``LIC''.
The authors would like to acknowledge and thank the Galileo Galilei Institute in Florence for their hospitality and support 
 for the period of the scientific program on
    ``Statistical Mechanics, Integrability and Combinatorics" during which part of this work was completed.}

\date{\today}

\begin{document}

\maketitle

\begin{abstract}
We consider the quantum Knizhnik--Zamolodchikov--Bernard equation for a face model with elliptic weights,
the SOS model. We provide explicit solutions as theta functions.
On the so-called combinatorial line, in which the model is equivalent
to the {\em three-colour model},
these solutions are shown to be eigenvectors of the transfer matrix with
periodic boundary conditions.
\end{abstract}

\section{Introduction}

The Knizhnik--Zamolodchikov (KZ) equations \cite{KZ} are a set of compatible
differential equations satisfied by conformal blocks in Conformal Field
Theory. The original work on KZ equations was concerned with the theory on a sphere,
but Bernard generalised it to the torus \cite{Bernard88a} 
and to higher genus Riemann surfaces \cite{Bernard88b},
leading to the so-called Knizhnik--Zamolodchikov--Bernard (KZB) equations.
These differential equations come from a flat connection
over the moduli space of Riemann surfaces with $L$ marked points; in the
present work, we shall
only be concerned with the variation of the marked
points, and not with the variation of the underlying Riemann surface,
whose genus will always be one.

The KZ equations are intimately related to the representation theory
of affine algebras.
In \cite{FR-qKZ}, {\em difference}\/ equations, now customarily
called qKZ equations, were introduced as analogues of the KZ equations
for quantised affine algebras. These equations were formulated in terms
of the $R$-matrix of the associated quantum integrable system,
which satisfies the Yang--Baxter equation. For a review of qKZ equations see 
\cite{Etingofbook} and references contained there.

Generalising both the KZB and qKZ equations, Felder
introduced the qKZB equations \cite{Felder-ICM}. An important new
ingredient was the use of an elliptic solution
of the {\em dynamical}\/ Yang--Baxter
equation. The qKZB equations can be viewed
as related to a representation of an (extended) affine Weyl
group, which will always be the affine symmetric group for us, 
on an appropriate functional space,
dynamical $R$-matrices being ``generalised $R$-matrices'' in the sense
of Cherednik \cite{Che-qKZ}. In the latter work, two types of solutions of the
qKZ(B) equation were
considered: ordinary solutions, corresponding to looking for
eigenvectors of the commutative subgroup of the affine Weyl group,
and symmetric solutions, which are invariant under the whole of the
affine Weyl group. In the second case, we call the set of equations
that these satisfy, the qKZ(B) {\em system}.

In fact, the qKZ system predated the qKZ equation -- see \cite{Smi}
where form factors of a quantum integrable model are shown to satisfy the qKZ system,
as well as \cite{JM-book} for a similar connection to correlation functions.
In the
context of dynamical $R$-matrices associated to so-called ABF models,
the qKZB system was considered in \cite{Fodaetal94} using the vertex operator approach \cite{JM-book}. This approach was developed further for 
a range of elliptic face and vertex models in \cite{MR1403536,MR1421016,MR1702389,MR1832074,MR1625187,MR2153660}. Our approach in this paper is different, and the rationale for 
developing this alternative to the vertex operator approach is discussed in Section 5.

In an unrelated development, Stroganov \cite{Strog-odd}
and Razumov and Stroganov \cite{RS-spin,RS-spintwist} noticed
that the ground state energy of the (odd size, or twisted even size) XXZ spin
chain at the value $\Delta=-1/2$ of its anisotropy parameter was
particularly simple.
The corresponding ground state entries 
turned out to have interesting connections to combinatorics
(see also \cite{RS-conj,BdGN-XXZ-ASM-PP}), and $\Delta=-1/2$
was dubbed the ``combinatorial point'' of the XXZ spin chain.

In an attempt to explain these ground state features,
Di Francesco and Zinn-Justin \cite{artic31} considered an inhomogeneous version
of certain models that are equivalent to the XXZ spin chain, namely a loop model
and the six-vertex model, and analysed 
the ground state entries as a function of the inhomogeneities
(or {\em spectral parameters}), showing in particular
that they formed polynomials in these spectral parameters. 
The equations satisfied by the inhomogeneous
ground state turned out to be nothing but
the qKZ system \cite{artic34}, and using related representation theory,
an explicit form of the ground state entries was provided
\cite{artic42,Tiago-thesis} in terms of certain contour integrals.
It is worth noting that the solution of the qKZ system makes
sense for any value of the anisotropy parameter $\Delta$, but only at
$\Delta=-1/2$ can it be interpreted as an eigenvector of the transfer matrix
of the model.

The XXZ and six-vertex models are based on a trigonometric solution
of the Yang--Baxter equation.
It is natural to try to generalise these results to an
{\em elliptic}\/ solution.
An obvious attempt is to consider the 
{\em XYZ}\/ and {\em eight-vertex}\/ models,
and indeed remarkable properties of the ground state of the XYZ/eight-vertex
model in odd size on a particular ``combinatorial line''
were observed in \cite{BM-8v,BM-P6a,BM-P6b,RS-XYZ}
and analysed using the qKZ(B) approach in
\cite{artic57}. This approach was partially successful in the sense that some
conjectures from \cite{BM-P6b} were proved, but no explicit formula was
found for the ground state entries.

Here we pursue a slightly different generalisation: we use the solution of
the dynamical Yang--Baxter equation which corresponds to SOS models \cite{ABF,DJMO1986}. 
We write the corresponding qKZB
system, and then construct some solutions. These solutions are shown to be
eigenvectors of the transfer matrix (with periodic boundary conditions) of the SOS model on the ``combinatorial line'', which in the SOS language corresponds to the three-colour model \cite{Baxter1970,Lieb1967}.
We note that the three-colour model was also recently studied
in \cite{rosengrenb},
but with different (Domain Wall) Boundary Conditions; experience with
the trigonometric limit suggests that there may be some connection between
this work and ours.


In slightly more detail, the results of the paper are as follows.
In Section~\ref{sec:sol}, we define the quantum integrable model (SOS model) and the associated qKZB system. We then show that a certain integral expression provides {\em two}\/ solutions of it. 
In Section~\ref{sec:anal}, we analyse in more detail these solutions, proving that they are {\em theta functions}\/ not only of the dynamical parameter, as is expected, but also of all spectral parameters. We show that they satisfy recurrence relations and a ``wheel condition''. In Section~\ref{sec:combi}, we restrict ourselves to the value $\eta=2\pi/3$ of the crossing parameter (the ``combinatorial line''), and show the eigenvector property for these solutions of qKZB. 
In Section~\ref{sec:disc}, we summarise our findings and discuss connections with existing work on qKZB equations.
Finally, Appendix~\ref{app:ex} provides the explicit example of size $L=4$, while Appendix~\ref{app:lim} describes the trigonometric and rational limit of our results and Appendix~\ref{app:graph} gives a pictorial representation of the qKZB system and equation.

\section{The Solution of the Difference Equations}\label{sec:sol}
In this section we introduce elliptic face weights, define difference equations (qKZB system) and then go on to construct particular
solutions. 
\subsection{The SOS weights}
We consider SOS weights which depend upon three complex parameters $\eta$, $\tau$ and $\zeta$ (where the elliptic nome $p$ is 
given by $p=e^{i\pi \tau}$),
spectral parameter $u\in \C$, and a height parameter $a\in \Z$ \cite{Baxter1973,DJMO1986}. Defining $f(u)=\vartheta_1(u|\tau)$ we have
\begin{align}
  \Wop[a][a\pm1][a\pm1][a\pm2][u] & = \frac{\f[\eta+u]}{\f[\eta-u]}, \nonumber \\
  \Wop[a][a\pm1][a\mp1][a][u] & = \frac{\f[u]\f[a\eta\pm\eta+\zeta]}{\f[\eta-u]\f[a\eta+\zeta]}, \label{eq:weights}\\
  \Wop[a][a\pm1][a\pm1][a][u] & = \frac{\f[\eta]\f[a\eta\mp u+\zeta]}{\f[\eta-u]\f[a\eta+\zeta]},\nonumber
\end{align}
It is conventional to encode SOS weights into a dynamical R-matrix ${\mathcal{R}}(u,a)\in\mbox{End}(\C^2\ot \C^2)$ \cite{FelVar1996}. Here, 
we adopt the alternative approach of considering the R-matrix as acting on a path space. More specifically, we consider 
a periodic SOS model on a  chain of $ L=2n$  sites
given in terms of the path Hilbert space $\mathcal{H}_{L}$ defined as the complex
span of the set $\mathcal{B}_{L}$ of basis vectors  given by
\begin{align*}
  \mathcal{B}_{L} & = \left\{\ket{a_{1},a_{2},\dots a_{L}}\left|\right. a_{i}\in \Z, \, |a_{i-1}-a_{i}|=1 \right\}
\end{align*}
with indices taken mod$( L)$.
In this case we define R-matrices $R_i(u)\in$End$(\mathcal{H}_L)$ by 
\begin{align}
  \bra{\bf{a}'} R_{i}(u) \ket{\bf{a}} & = \Wop[a_{i-1}][a_{i}'][a_{i}][a_{i+1}][u] \prod_{j\neq i} \delta_{a_{j}}^{a_{j}'}, \quad
\hbox{where}\quad   \ket{\bf{a}}=|a_1,a_2,\cdots,a_L\rangle. \label{eq:Rmat}
\end{align}
With this definition the face-version of the Yang-Baxter equation (which \eqref{eq:weights} obey) is written simply as the standard
\begin{align}
  R_{i}(u)R_{i+1}(u+v)R_{i}(v) & = R_{i+1}(v)R_{i}(u+v)R_{i+1}(u), \label{eqnYBE} 
\end{align}
The SOS transfer matrix $t(u)\in \mbox{End}(\mathcal{H}_L)$ is then defined by
\begin{align}
  \bra{\bf{a}'} t(u) \ket{\bf{a}} & = \prod_{i=1}^{L} \Wop[a_{i-1}'][a_{i}'][a_{i-1}][a_{i}][u-u_{i}]. \label{eq:tmat}
\end{align}

\subsection{The qKZB System}
In this section we define the qKZB system associated with our elliptic face weights. 
First of all we define $\rho:\cB_L\rightarrow \cB_L$
\begin{align*}
  \rho\,\ket{a_1,\cdots,a_{L-1},a_L} & = \ket{a_L,a_1,\cdots,a_{L-1}}. \\ 
\end{align*}
We then have 
\begin{dfn}[Level-$\ell$ qKZB System]
  Consider a function $\Psi: \C^L\rightarrow \cH_L$. The level-$\ell$ qKZB system is defined as the set of equations  
\begin{align}
  R_i(u_i-u_{i+1}) \Psi(u_1,u_2,\cdots,u_i,u_{i+1},\cdots,u_L) & =  \Psi(u_1,u_2,\cdots,u_{i+1},u_{i},\cdots,u_L)\label{qkz1} \\
  \rho \, \Psi(u_1,u_2,\cdots,\cdots,u_{L-1},u_L) & = \kappa\,\Psi(u_L+ s ,u_1,u_2,\cdots,u_{L-1})\label{qkz2}
\end{align}
where $s= (\ell+2) \eta$ and $\kappa\in \C^{\times}$ is a constant.
This pair of equations implies the following: 
\begin{align}
  S_i \Psi(u_1,u_2,\cdots,u_i,u_{i+1},\cdots,u_L) & = \kappa \Psi(u_1,u_2,\cdots,u_i+s,u_{i+1},\cdots,u_L)\label{qkzeq}
\end{align}
where
\begin{align*}
  S_i & :=  \left[R_{i-1}(u_{i}-u_{i-1}+s) \cdots R_{1}(u_{i}-u_{1}+s)\right]\rho \left[R_{L -1}(u_{i}-u_{L}) \cdots R_{i}(u_{i}-u_{i+1})\right].
\end{align*}
Note, that we shall refer to \eqref{qkzeq} itself referred as the qKZB equation, as opposed to qKZB system. Pictorial representations 
of the qKZB system and qKZB equation are  given in Appendix~\ref{app:graph}.
\end{dfn}

We define components $\Psi_{\ba}(\bu)\in \C$ of a function $\Psi: \C^L\rightarrow \cH_L$ by 
\begin{align*}
  \Psi(\bu) = \sli_{\ket{\ba} \in \mathcal{B}_L} \Psi_{\ba}(\bu)\ket{\ba}, \quad\hbox{where}\quad \ket{\ba}=\ket{a_1,a_2,\cdots,a_L}.
\end{align*}
The qKZB equation can be written in terms of the coefficients as
\begin{align}
  & \Psi_{\ba'}(\ldots,u_i+s,\ldots) \label{qkzeqcoeff}\\
  & =  \kappa^{-1} \sum_{\ba} \Psi_{\ba}(\ldots)\, \delta_{a_{i}'}^{a_{i-1}} \prod_{j=1}^{i-1} \Wop[a_{j-1}'][a_{j}'][a_{j-1}][a_{j}][u_{i}-u_{j}+s] \prod_{j=i+1}^{L} \Wop[a_{j-1}'][a_{j}'][a_{j-1}][a_{j}][u_{i}-u_{j}]. \nonumber
\end{align}

\subsection{The Solution of the level-1 qKZB System}
The main result of this section is given by Theorem \ref{thm:qkz}. We shall present a series of definitions and Lemmas that lead to this
result. 
For a given $\ket{\ba} \in \mathcal{B}_L$
let us define a vector $\alpha$ (vector $\balpha$) in terms of the $n$ positions where height variables decrease (increase) to the left. Namely, we define
\begin{align*}
  \alpha&=(\alpha_1,\alpha_2,\cdots,\alpha_n),\quad \hbox{where}\quad \alpha_i\in\{1,2,\cdots, L\},\quad 
  \alpha_i<\alpha_{i+1}, \quad \hbox{and}\quad a_{\alpha_i}-a_{\alpha_i-1}=+1,\\
  \balpha&=(\balpha_1,\balpha_2,\cdots,\balpha_n),\quad \hbox{where}\quad \balpha_i\in\{1,2,\cdots, L\},\quad 
  \balpha_i<\balpha_{i+1}, \quad \hbox{and}\quad a_{\balpha_i}-a_{\balpha_i-1}=-1.
\end{align*}
For example, for $ L=4$ and path $\ba=(1,2,3,2)$ we have $\alpha=(2,3)$,  $\balpha=(1,4)$; and for $\ba=(3,2,3,2)$, we have $\alpha=(1,3)$, $\balpha=(2,4)$.

We will first proceed in steps to define and consider the properties of two different functions $\Psi^{(j)}(\bu)$ $(j=2,3)$ expressed as
integrals. These are the functions that will be shown to satisfy the qKZB system.
We show also define two alternative integral expressions $\bPsi^{(j)}(\bu)$ $(j=2,3)$. It will turn out that $\Psi^{(j)}(\bu)=
\bPsi^{(j)}(\bu)$, but the alternative form will prove useful in the proof of the cyclicity relation \eqref{qkz2}.

First of all, we define the relevant integrands. 
\begin{dfn}\label{def:Idef}
  With $ L=2n$, we define the four functions $I^{(2)},I^{(3)},\bI^{(2)},\bI^{(3)} : \C^{2n}\times \C^{n} \rightarrow \cH_L$ by their components
\begin{align}
  & I_\ba^{(j)}(u_1,\cdots,u_{L}|v_1,\cdots,v_n):= \f[\eta]^n \f[a_L \eta +\zeta] \fj[a_{L}\eta+\zeta -n\eta+2\sum_{l=1}^{n}v_{l}- \sum_{m=1}^{L}u_{m}] \label{eq:Idef} \\
  & \times \left\{ \prod_{l=1}^{n} \frac{\f[a_{\alpha_{l}}\eta+\zeta-v_{l}+u_{\alpha_{l}}] \left( \prod_{l < m \leq n} \f[v_{l}-v_{m}]\f[\eta-v_{l}+v_{m}]\right)}{\left(\prod_{1 \leq m \leq \alpha_{l}} \f[u_{m}-v_{l}]\right)\left(\prod_{\alpha_{l} \leq m \leq  L} \f[\eta-v_{l}+u_{m}]\right)} \right\}, \nonumber \\
  & \bI_\ba^{(j)}(u_1,\cdots,u_{L}|v_1,\cdots,v_n):= \f[\eta]^n \f[a_L \eta +\zeta] \fj[a_{L}\eta+\zeta +n\eta-2\sum_{l=1}^{n}v_{l}+ \sum_{m=1}^{L}u_{m}] \label{eq:bIdef} \\
  & \times \left\{ \prod_{l=1}^{n} \frac{\f[a_{\balpha_{l}}\eta+\zeta+v_{l}-u_{\balpha_{l}}] \left( \prod_{l < m \leq n} \f[v_{l}-v_{m}] \f[\eta-v_{l}+v_{m}]\right)}{\left(\prod_{1 \leq m \leq \balpha_{l}} \f[u_{m}-v_{l}]\right)\left(\prod_{\balpha_{l} \leq m \leq  L} \f[\eta-v_{l}+u_{m}]\right)} \right\}, \nonumber
\end{align}
where $f^{(j)}(u)=\vartheta_{j}(u|2\tau)$.
\end{dfn}

We shall use these functions as the integrands of multiple integral expressions for solutions of the qKZB system. To this end it is useful to first 
show the following:

\begin{lem}\label{lem:per}
  The functions $I_\ba^{(j)}(\bu|\bv)$, $\bI_\ba^{(j)}(\bu|\bv)$ are meromorphic functions of each of the arguments $(v_1,v_2,\cdots,v_n)=:\bv$, and are doubly periodic with
\begin{align*}
  I_\ba^{(j)}(\bu|\bv) & = I_\ba^{(j)}(\bu|v_1,v_2,\cdots,v_i+(r_1+r_2\tau)\pi,\cdots,v_n),\quad r_1,r_2\in\Z \\
  \bI_\ba^{(j)}(\bu|\bv)& = \bI_\ba^{(j)}(\bu|v_1,v_2,\cdots,v_i+(r_1+r_2\tau)\pi,\cdots,v_n),\quad r_1,r_2\in\Z \\
\end{align*}

\end{lem}
\begin{proof} Meromorphicity comes directly from the properties of the Jacobi theta functions. The double periodicity comes from the relations
\begin{align}
  \vartheta_{1}(u+(r_1+r_2\tau)\pi|\tau) & = (-1)^{r_1+r_2} p^{-r_2^2} e^{-2ir_{2}u} \vartheta_{1}(u|\tau)\nonumber \\
  \vartheta_{2}(u+(r_1+r_2\tau)\pi|\tau) & = (-1)^{r_1} p^{-r_2^2} e^{-2ir_{2}u}\vartheta_{2}(u|\tau) \label{eq:periodic}\\
  \vartheta_{3}(u+(r_1+r_2\tau)\pi|\tau) & = p^{-r_2^2} e^{-2ir_{2}u} \vartheta_{3}(u|\tau) \nonumber 
\end{align}
with quasi-periodic prefactors cancelling from the top and bottom of \eqref{eq:Idef}.
\end{proof}

Now we define functions that will ultimately be shown to satisfy the qKZB system: 
\begin{dfn}\label{def:Psi}
  We define the functions $\Psi^{(j)}, \bPsi^{(j)}: \C^L \rightarrow \cH_L$ $(j=2,3)$ in terms of components $\Psi^{(j)}_\ba(\bu), \bPsi^{(j)}_\ba(\bu)$ by
\begin{align}
  & \Psi_{\ba}^{(j)}(\bu)= \Phi_{0}(\bu) \frac{1}{c^n} \int_{C_\alpha^{(1)} } \frac{dv_1}{2\pi i}  \int_{C_\alpha^{(2)}} \frac{dv_2}{2\pi i}  \cdots \int_{C_\alpha^{(n)}} \frac{dv_n}{2\pi i} \; I_\ba^{(j)}(\bv|\bu)\label{eq:intform1}\\
  & \bPsi_{\ba}^{(j)}(\bu)= \Phi_{0}(\bu) \frac{1}{c^n} \int_{\bC ^{(1)}_{\balpha}  } \frac{dv_1}{2\pi i}  \int_{\bC_{\balpha}^{(2)}} \frac{dv_2}{2\pi i}  \cdots \int_{\bC_{\balpha}^{(n)}} \frac{dv_n}{2\pi i} \; \bI_\ba^{(j)}(\bv|\bu)\label{eq:intform2}\\
  & \hbox{where}\quad \Phi(\bu) = \prod_{1 \leq i < j \leq  L} \f[\eta-u_{i}+u_{j}], \quad c=\oint \frac{dz}{2\pi i \f[z]}=
\frac{1}{2q^{1/4} (q^2;q^2)_\infty^3}, \nonumber 
\end{align}
and for a given $\alpha$, the contours are defined such that:\\[-2mm]

\hspace*{3mm} The $v_\ell$ contour $C^{(\ell)}_\alpha$ circles only the poles at $u_m$ for $m=1,2,\cdots,\alpha_\ell$,

\hspace*{3mm} The $v_\ell$ contour  $\bC^{(\ell)}_\alpha$ circles only the poles at  $u_m+\eta$ for $m=\alpha_\ell,\alpha_{\ell}+1,\cdots, L$.

\end{dfn} 
A useful result relating integrals over the two contours is the following:
\begin{lem}\label{lem:cbar}
\begin{align*}
  \int_{\bC^{(\ell)}_\alpha } \frac{dv_\ell}{2\pi i}  I_\ba^{(j)}(\bv|\bu) = - \int_{C^{(\ell)}_\alpha } \frac{dv_\ell}{2\pi i}  I_\ba^{(j)}(\bv|\bu),\quad
  \int_{\bC^{(\ell)}_{\balpha} } \frac{dv_\ell}{2\pi i}  \bI_\ba^{(j)}(\bv|\bu) = - \int_{C^{(\ell)}_{\balpha} } \frac{dv_\ell}{2\pi i}  \bI_\ba^{(j)}(\bv|\bu)  
\end{align*}
\end{lem}
\begin{proof}
We consider the case  of $I_\ba^{(j)}(\bv|\bu) $: From Lemma \ref{lem:per} we know that 
\begin{align*}
  \int_{D_\tau } \frac{dv_\ell}{2\pi i}  I_\ba^{(j)}(\bv|\bu) & = 0
\end{align*}
where $D_\tau$ is the boundary of the parallelogram in the complex plane with vertices
\begin{align*}
  \frac{\pi}{2}(-1-\tau),\quad \frac{\pi}{2}(1-\tau), \quad \frac{\pi}{2}(-1+\tau), \quad \frac{\pi}{2}(1+\tau)
\end{align*}
($D_\tau$ can be translated if necessary to avoid poles on the boundary).
The catalogue of poles occurring inside the domain $D_\tau$ consists of those at $u_m$ ($m=1,2,\cdots,\alpha_\ell$) and
$u_m+\eta$ ($m=\alpha_\ell,\alpha_{\ell}+1,\cdots, L$) or double-periodically translated versions of them.
Hence we have
\begin{align*}
  \int_{D_\tau } \frac{dv_\ell}{2\pi i}  I_\ba^{(j)}(\bv|\bu)=\int_{\C^{(\ell)} } \frac{dv_\ell}{2\pi i}  I_\ba^{(j)}(\bv|\bu) +\int_{\bC^{(\ell)} } \frac{dv_\ell}{2\pi i}  I_\ba^{(j)}(\bv|\bu) = 0.
\end{align*}
The proof of the statement of involving the integral of  $\bI_\ba^{(j)}(\bv|\bu) $ is similar. 
\end{proof}

\noindent
For convenience we define $\tau_{i,i+1}$, which acts on functions of $L$ variables, by
\begin{align*}
  \tau_{i,i+1}\, g(z_1,\cdots,z_i,z_{i+1},\cdots,z_L) & = g(z_1,\cdots,z_{i-1},z_{i+1},z_{i},z_{i+2},\cdots,z_L)\\
\end{align*}
Now we have a lemma regarding the solution of the first equation \eqref{qkz1} of the qKZB system:
\begin{lem}\label{lem:qkz1}
The functions $\Psi^{(j)}(\bu)$ and $\bPsi^{(j)}(\bu)$ $(j=2,3)$ satisfy the following special cases of the exchange relation \eqref{qkz1}:
\begin{align*}
  & (i) && \tau_{i,i+1} \Psi^{(j)}_{\dots a,a-1,a-2\dots}(\bu)  =  \frac{\f[\eta+u_{i}-u_{i+1}]}{\f[\eta-u_{i}+u_{i+1}]} \Psi^{(j)}_{\dots a,a-1,a-2\dots}(\bu), \\
  & (ii) && \tau_{i,i+1} \bPsi^{(j)}_{\dots a,a+1,a+2\dots}(\bu)  =  \frac{\f[\eta+u_{i}-u_{i+1}]}{\f[\eta-u_{i}+u_{i+1}]} \bPsi^{(j)}_{\dots a,a+1,a+2\dots}(\bu), \\
  & (iii) &&  \Psi_{\dots a,a-1,a\dots}(\bu)  \\
  &&& = \frac{\f[a\eta+\zeta]\f[\eta-u_{i}+u_{i+1}]\tau_{i,i+1} - \f[\eta]\f[a\eta+\zeta-u_{i}+u_{i+1}]}{\f[u_{i}-u_{i+1}]\f[a\eta+\eta+\zeta]} \Psi_{\dots a,a+1,a\dots}(\bu), \nonumber\\
  & (iv) && \Psi_{\dots a,a+1,a\dots}(\bu)\\
  &&& = \frac{\f[a\eta+\zeta]\f[\eta-u_{i}+u_{i+1}]\tau_{i,i+1} - \f[\eta]\f[a\eta+\zeta+u_{i}-u_{i+1}]}{\f[u_{i}-u_{i+1}]\f[a\eta-\eta+\zeta]} \Psi_{\dots a,a-1,a\dots}(\bu), \nonumber
\end{align*}
for $\Psi=\Psi^{(j)}, \bPsi^{(j)}$ and $j=2,3$.
\end{lem}
\begin{proof} 
The fact that the first two relations, (i) and (ii), are satisfied follows directly from the following three identities:
\begin{align*}
  \tau_{i,i+1} I^{(j)}_{(\dots a,a-1,a-2\dots)}(\bv|\bu) & = I^{(j)}_{(\dots a,a-1,a-2\dots)}(\bv|\bu)\\
  \tau_{i,i+1} \bar{I}^{(j)}_{(\dots a,a+1,a+2\dots)}(\bv|\bu) & = \bar{I}^{(j)}_{(\dots a,a+1,a+2\dots)}(\bv|\bu),\\
  \tau_{i,i+1} \Phi(\bu) & = \frac{f(\eta-u_{i+1}+u_{i})}{f(\eta-u_{i}+u_{i+1})} \Phi(\bu)
\end{align*}
Moving on to statements (iii) and (iv) we provide the proof only for $\Psi^{(j)}(\bu)$ - the $\bPsi^{(j)}(\bu)$ case is very similar.
From Eq. \mref{eq:Idef} that we have 
\begin{align*}
  \tau_{i(i+1)} I^{(j)}_{\dots a,a+1,a\dots}(\bv|\bu) & = \frac{\f[u_{i}-v_{l}]\f[\eta(a+1)+\zeta-v_{l}+u_{i+1}]}{\f[u_{i+1}-v_{l}]\f[\eta(a+1)+\zeta-v_{l}+u_{i}]} I^{(j)}_{\dots a,a+1,a\dots}(\bv|\bu), \\
  \tau_{i(i+1)} I^{(j)}_{\dots a,a-1,a\dots}(\bv|\bu) & = \frac{\f[\eta-v_{l}+u_{i+1}]\f[a\eta+\zeta-v_{l}+u_{i}]}{\f[\eta-v_{l}+u_{i}]\f[a\eta+\zeta-v_{l}+u_{i+1}]} I^{(j)}_{\dots a,a-1,a\dots}(\bv|\bu), \\
  I^{(j)}_{\dots a,a-1,a\dots}(\bv|\bu) & = \frac{\f[\eta-v_{l}+u_{i}]\f[a\eta+\zeta-v_{l}+u_{i+1}]}{\f[u_{i+1}-v_{l}]\f[\eta(a+1)+\zeta-v_{l}+u_{i}]} I^{(j)}_{\dots a,a+1,a\dots}(\bv|\bu),
\end{align*}

\noindent It follows from these relations that we have 
\begin{align}
  &   \f[a\eta + \zeta]\f[\eta-u_{i}+u_{i+1}] \tau_{i(i+1)}\left[\Phi_{0}(\bu) I^{(j)}_{\dots a,a+1,a\dots}(\bv|\bu)\right] \label{eq:ex2}  \\
  & - \f[u_{i}-u_{i+1}] \f[a\eta+\eta+\zeta] \left[\Phi_{0}(\bu I^{(j)})_{\dots a,a-1,a\dots}(\bv|\bu)\right] \nonumber \\
  & - \f[\eta]\f[a\eta+\zeta-u_{i}+u_{i+1}] \left[\Phi_{0}(\bu)  I^{(j)}_{\dots a,a+1,a\dots}(\bv|\bu)\right] \nonumber \\
  & = \, \frac{\Xi \left( (a+1)\eta+\zeta-v_{l}+u_{i+1},a\eta+\zeta,\eta-u_{i+1}+u_{i},u_{i}-v_{l}\right)}{\f[ u_{i+1}-v_{l}] \f[(a+1)\eta+\zeta-v_{l}+u_{i}]} 
  \left[ \Phi_{x0}(\bu) I^{(j)}_{\dots a,a+1,a\dots}(\bv|\bu)\right] \nonumber \\[3mm]
  &   \f[a\eta+\zeta]\f[\eta-u_{i}+u_{i+1}] \tau_{i(i+1)}\left[\Phi_{0}(\bu) I^{(j)}_{\dots a,a-1,a\dots}(\bv|\bu)\right] \label{eq:ex3} \\
  & - \f[u_{i}-u_{i+1}]\f[a\eta-\eta+\zeta]\left[\Phi_{0}(\bu) I^{(j)}_{\dots a,a+1,a\dots}(\bv|\bu)\right] \nonumber \\
  & - \f[\eta]\f[a\eta+\zeta+u_{i}-u_{i+1}] \left[\Phi_{0}(\bu) I^{(j)}_{\dots a,a-1,a\dots}(\bv|\bu)\right] \nonumber \\
  & = \frac{\Xi(a\eta+\zeta-v_{l}+u_{i}, a\eta+\zeta, \eta-u_{i+1}+u_{i}, \eta-v_{l}+u_{i+1})}{\f[\eta-v_{l}+u_{i}]\f[a\eta+\zeta-v_{l}+u_{i+1}]} \left[\Phi_{0}(\bu) I^{(j)}_{\dots a,a-1,a\dots}(\bv|\bu)\right] \nonumber 
\end{align}
where the function $\Xi(z_1,z_2,z_3,z_4)$ is defined by
\begin{align*}
  & \Xi(2z_1,2z_2,2z_3,2z_4):= \f[2z_1] \f[2z_2] \f[2z_3] \f[2z_4] \\
  & - \f[-z_1+z_2+z_3+z_4] \f[z_1-z_2+z_3+z_4] \f[z_1+z_2-z_3+z_4] \f[z_1+z_2+z_3-z_4] \\
  & - \f[z_1+z_2+z_3+z_4] \f[z_1-z_2-z_3+z_4] \f[z_1-z_2+z_3-z_4] \f[z_1+z_2-z_3-z_4].
\end{align*}
However, a Riemann identity for $\vartheta_{1}$ states that $\Xi(z_1,z_2,z_3,z_4)=0$, and hence \eqref{eq:ex2} and \eqref{eq:ex3} imply 
(iii) and (iv) respectively. 
\end{proof}

We note that the only cases of the exchange relation \eqref{qkz1} not covered by the above lemma are relations (i) and (ii) with $\Psi^{(j)}(\bu)$ and $\bPsi^{(j)}(\bu)$ interchanged, which will follow once we show the equality between the two. However, to do this it is useful
to have an integrated form of the expressions \eqref{eq:intform1} and \eqref{eq:intform2} for $\Psi^{(j)}(\bu)$ and $\bPsi^{(j)}(\bu)$. 
Such integrated expressions can be expressed in terms of the following sets:

\begin{dfn}
We define the sets
\begin{align*}
  S_{\alpha}   & = \{ (i_{1},\dots i_{n}) | 1 \leq i_{l} \leq \alpha_{l}, \, i_{l}\neq i_{m},\, 1 \leq l,m \leq n\}, \\
  \bS_{\alpha} & = \{ (i_{1},\dots i_{n}) | \alpha_{l} \leq i_{l} \leq  L\, \, i_{l}\neq i_{m},\, 1 \leq l,m \leq n\},
\end{align*}
where $\alpha$ can be replaced with $\balpha$. For example, for $\ba=(3,4,5,4)$ we have $\alpha=(2,3)$, $\balpha=(1,4)$, and
\begin{align*}
  S_{\alpha}=\{(1,2),(1,3),(2,1),(2,3)\}, \quad  \bS_{\balpha}=\{(1,4),(2,4),(3,4)\}.
\end{align*}
\end{dfn}

\noindent The set $S_{\alpha}$ contains the possible $(i_1,\cdots,i_n)$ at which the multiple residue at 
$\bv=(u_{i_1},u_{i_2},\cdots, u_{i_n})$ is non-zero. Similarly, $\bS_{\balpha}$ contains the possible $(i_1,\cdots,i_n)$ at which the multiple residue at 
$\bv=(u_{i_1}+\eta,u_{i_2}+\eta,\cdots, u_{i_n}+\eta)$ is non-zero. It follows that we have the integrated expressions
\begin{align}
  \Psi^{(j)}_{\ba}(\bu) & = \left[\f[\eta]\right]^{n} \f[a_{L}\eta+\zeta] \sum_{{\bf i}\in S_{\alpha}} \psi^{(j)}_{\ba,{\bf i}}(\bu),
  \quad \bPsi^{(j)}_{\ba}(\bu) = \left[\f[\eta]\right]^{n} \f[a_{L}\eta+\zeta] \sum_{\bi \in \bS_{\balpha}} \bpsi^{(j)}_{\ba,\tilde{i}}(\bu), \label{eq:intexp} \\
  \psi^{(j)}_{\ba,{\bf i}}(\bu) & = \fj[a_{L}\eta+\zeta -n\eta+2\sum_{l=1}^{n}u_{i_{l}}- \sum_{m=1}^{L}u_{m}] \left\{ \frac{\prod_{1 \leq l < m \leq n} \f[u_{i_{l}}-u_{i_{m}}]}{\prod_{l=1}^{n}\prod_{1 \leq m \leq \alpha_{l},\, m\neq i_{l}} \f[u_{m}-u_{i_{l}}]} \right\}  \nonumber \\
  & \quad \times  \prod_{l=1}^{n} \f[a_{\alpha_{l}}\eta+\zeta-u_{i_{l}}+u_{\alpha_{l}}] \left\{ 	\frac{\prod_{1 \leq i < j \leq  L} \f[\eta-u_{i}+u_{j}]\prod_{1 \leq l < m \leq n} \f[\eta-u_{i_{l}}+u_{i_{m}}]}{\prod_{l=1}^{n} \prod_{\alpha_{l} \leq m \leq  L} \f[\eta-u_{i_{l}}+u_{m}]} \right\}, \nonumber \\
  \bpsi^{(j)}_{\ba,\bi}(\bu)
  & = \fj[a_{L}\eta+\zeta -n\eta-2\sum_{l=1}^{n}u_{i_{l}} +\sum_{m=1}^{L}u_{m}] \left\{\frac{\prod_{1 \leq l < m \leq n} \f[u_{i_{l}}-u_{i_{m}}]}{\prod_{l=1}^{n}\prod_{\balpha_{l} \leq m \leq  L,\,m\neq i_{l}} \f[u_{i_{l}}-u_{m}]} \right\} \nonumber \\
  & \quad\times \prod_{l=1}^{n}  \f[a_{\balpha_{l}}\eta+\eta+\zeta+u_{i_{l}}-u_{\balpha_{l}}]   \left\{\frac{\prod_{1 \leq i < j \leq  L} \f[\eta-u_{i}+u_{j}]\prod_{1 \leq l < m \leq n} \f[\eta-u_{i_{l}}+u_{i_{m}}]}{\prod_{l=1}^{n} \prod_{1 \leq m \leq \balpha_{l}} \f[\eta-u_{m}+u_{i_{l}}]} \right\}. \nonumber
\end{align}

Now we can demonstrate the following:
\begin{lem}\label{lem:psibpsi}
  We have $\Psi^{(j)}(\bu)=\bPsi^{(j)}(\bu)$ $(j=2,3)$.
\end{lem}
\begin{proof}
  The strategy for proving the equality is to demonstrate equality for a particular choice of $\ba$, and then to use the exchange relations following from
 Lemma \ref{lem:qkz1} to show that the relation follows for all $\ba$. Note that for the choice $\ba=(a+1,a+2,\cdots,a+n,a+n-1,\cdots,a+1,a)$ we have 
\begin{align*}
  \alpha & = (1,2,\cdots,n) & S_{\alpha}&=\{(1,2,\cdots,n)\} \\
  \balpha & =(n+1,n+2,\cdots,2n),\quad & \bS_{\balpha} & = \{(n+1,n+2,\cdots,2n)\}.
\end{align*}
That is, both sets contain a single element. Using \eqref{eq:intexp} we have
\begin{align}
  \Psi^{(j)}_{a+1,\dots, a+n, \dots, a+1,a}(\bu) 
  & = \bPsi^{(j)}_{a+1,\dots, a+n, \dots, a+1,a}(\bu) \label{qkzupstate}\\
  & = \fj[a\eta+\zeta -n\eta +\sum_{m=1}^{n}u_{m} -\sum_{m=n+1}^{2n}u_{m}] \prod_{l=0}^{n} \f[a\eta+ l\eta+\zeta] \nonumber \\
  & \quad \times \prod_{1\leq i < j \leq n} \f[\eta-u_{i}+u_{j}] \prod_{n+1 \leq i < j \leq  L} \f[\eta-u_{i}+u_{j}] \nonumber
\end{align}
All other components $\Psi^{(j)}(\bu)$ can then be obtained uniquely in terms of $\Psi^{(j)}_{a+1,\dots, a+n, \dots, a+1,a}(\bu)$ by using the exchange relation
(iii) and (iv) from Lemma \ref{lem:qkz1} (and similar for $\bPsi^{(j)}(\bu)$) proving the equality.
\end{proof}

Finally, we show that $\Psi^{(j)}(\bu)$ satisfy the second qKZB system equation \eqref{qkz2}:
\begin{lem} \label{lem:qkz2} The functions $\Psi^{(j)}(\bu)$ satisfy the condition 
\begin{align}
  \Psi^{(j)}_{a_{L}a_{1}...a_{L-1}}(u_{1},u_{2},\dots,u_{L}) & = -\Psi^{(j)}_{a_{1}...a_{L-1}a_{L}}(u_{2},\dots,u_{L},u_{1} - 3\eta)
\label{eq:cycle}
\end{align}
\end{lem}
\begin{proof}
We need to split our calculation into two cases depending upon the value of $(a_{L}-a_{L-1})$. We first consider the possibility $a_{L}-a_{L-1}=1$. On the left hand side of \eqref{eq:cycle}  we have $\alpha_1=1$ and using the integral 
expression \eqref{eq:intform1}, the contour $C^{(1)}$ circles just one pole at $v_{1}= u_{1}$.
Thus we can compute
\begin{align*}
  & \Psi^{(j)}_{a_{L}a_{1}...a_{L-1}}(u_{1},u_{2},\dots,u_{L}) \\
  & = \left[\f[\eta]\right]^{n-1} \f[(a_{L}-1)\eta+\zeta] \f[a_{L}\eta+\zeta] \left[\prod_{2 \leq i < j \leq  L} \f[\eta-u_{i}+u_{j}]\right]\\
  & \quad \times \oint_{C^{(2)}_\alpha} \cdots \oint_{C^{(n)}_\alpha} \left(\prod_{l=2}^{n} \frac{dv_{l}}{2\pi i c}\right) \fj[(a_{L}-1)\eta+\zeta -n\eta +u_{1}+ 2\sum_{l=2}^{n}v_{l}- \sum_{m=2}^{L}u_{m}] \\
  & \quad \times \left\{ \prod_{l=2}^{n}  
  \frac{\f[a_{\alpha_{l-1}}\eta+\zeta-v_{l}+u_{\alpha_{l-1}+1}] \f[\eta-u_{1}+v_{l}] \left( \prod_{l < m \leq n} \f[v_{l}-v_{m}] \f[\eta-v_{l}+v_{m}]\right)}{\left(\prod_{2 \leq m \leq \alpha_{l-1}+1} \f[u_{m}-v_{l}]\right)\left(\prod_{\alpha_{l-1}+1 \leq m \leq  L} \f[\eta-v_{l}+u_{m}]\right)} \right\}
\end{align*}

\noindent
To compute $\Psi^{(j)}_{a_{1}...a_{L-1}a_{L}}(u_{2},\dots,u_{L},u_{1} - 3\eta)$, i.e., the expression on the right-hand-side of Eq. \mref{eq:cycle}, we note that the new vector $\alpha'$ has  $\alpha'_{n}= L$. The desired integral of $v_l$ over the corresponding  $C^{(n)}_{\alpha'}$ may be rewritten using Lemma \ref{lem:cbar} as minus the integral over  $\bC^{(n)}_{\alpha'}$, and this latter integral is simple in that it circles the single pole at $v_n=u_{1}-2\eta$.
Performing this integral gives
\begin{align*}
  & \Psi^{(j)}_{a_{1}...a_{L-1}a_{L}}(u_{2},\dots,u_{L},u_{1} - 3\eta) \\
  & = -\left[\f[\eta]\right]^{n-1} \f[a_{L}\eta+\zeta] \f[(a_{L}-1)\eta+\zeta] \left[\prod_{1 \leq i < j \leq  L-1} \f[\eta-u_{i+1}+u_{j+1}]\right]  \\
  & \quad \times \oint_{C^{(1)}_{\alpha'}} \cdots \oint_{C^{(n-1)}_{\alpha'}} \left(\prod_{l=1}^{n-1} \frac{dv_{l}}{c_{f}}\right) \fj[(a_{L}-1)\eta+\zeta -n\eta + u_{1}+ 2\sum_{l=1}^{n-1}v_{l}- \sum_{m=2}^{L}u_{m}] \\
  & \quad \times \left\{ \prod_{l=1}^{n-1}  
  \frac{\f[a_{\alpha_{l}}\eta+\zeta-v_{l}+u_{\alpha_{l}+1}] \f[\eta-u_{1}+v_{l}] \left( \prod_{l < m \leq n-1} \f[v_{l}-v_{m}] \f[\eta-v_{l}+v_{m}]\right)}{\left(\prod_{1 \leq m \leq \alpha_{l}} \f[u_{m+1}-v_{l}]\right)\left(\prod_{\alpha_{l} \leq m \leq  L-1} \f[\eta-v_{l}+u_{m+1}]\right)} \right\}
\end{align*}
Changing the integration variables we obtain the desired equality . 

In order to treat the case $a_L-a_{L-1}=-1$ we use a very similar argument, the key difference being that we start from the second integral
expression \eqref{eq:intform2} and use the equality $\Psi^{(j)}(\bu)=\bPsi^{(j)}(\bu)$ given by Lemma \ref{lem:psibpsi}. 
\end{proof}

Lemmas \ref{lem:qkz1}, \ref{lem:psibpsi} and \ref{lem:qkz2} lead to the main theorem:
\begin{thm} \label{thm:qkz} The functions $\Psi^{(2)}(\bu)$ and $\Psi^{(3)}(\bu)$ are solutions of the qKZB system \eqref{qkz1} and \eqref{qkz2} with level $\ell=1$ and $\kappa=-1$. 
\end{thm}
\begin{proof}
As $\Psi^{(j)}(\bu)=\bPsi^{(j)}(\bu)$ (Lemma \ref{lem:psibpsi}) we find that relations proven in Lemma \ref{lem:qkz1}, i.e. (i)-(iv), cover all cases of the exchange relations. Hence $\Psi^{(2)}(\bu)$ and $\Psi^{(3)}(\bu)$ satisfy Eq. \eqref{qkz1}. On the other hand Lemma \ref{lem:qkz2} shows that $\Psi^{(2)}(\bu)$ and $\Psi^{(3)}(\bu)$ satisfy Eq. \eqref{qkz2} with level $\ell=1$ and $\kappa=-1$. 
\end{proof}

Finally, by way of example, the explicit form of the integrated form of the solution $\Psi^{(3)}(u)$ is given in Appendix~\ref{app:ex}. 

\section{Analytic properties of the solutions} \label{sec:anal}
The properties of $\Psi^{(j)}$ as functions of $u_1,\ldots,u_L$ and of $\zeta$
are described in this section. In fact, {\em in this section only},
it will be convenient to write $u_0:=\zeta$ to treat all variables
simultaneously.

\subsection{Pseudo-periodicity}
We first have the following general result:
\begin{lem} \label{lem:pseudoperiod}
Given a vector-valued function $\Psi=\sum_{\mathbf a}\Psi_{\mathbf a}\ket{\mathbf a}$ which is analytic in the
variables $(u_0,u_1,\ldots,u_L)$,
and an element $(\lambda_0=\pi(m_0+n_0\tau),
\ldots,\lambda_L=\pi(m_L+n_L\tau))\in\Lambda$,
$\Psi$ is a
solution of the $q$KZ equation \eqref{qkzeq} at level $1$ ($s=3\eta$) iff 
$\Psi'$ is,
where $\Psi'=\sum_{\mathbf a}\Psi'_{\mathbf a}\ket{\mathbf a}$ 
is defined by any of the following:
\begin{enumerate}
\item $\Psi'_{\mathbf a}(u_0,\ldots,u_L)=\Psi(u_0,\ldots,u_m+\pi,\ldots,u_L)$
for some $m=0,1,\ldots,L$.
\item
$\Psi'_{\mathbf a}(u_0,\ldots,u_L)=\Psi_{\mathbf a}(u_0+\pi\tau,\ldots,u_L) e^{i\left(3(a_L\eta+\zeta)+\sum_{j=1}^L (a_j\eta+\zeta)+\sum_{j=1}^L (a_j-a_{j-1})u_j\right)}$.
\item
$\Psi'_{\mathbf a}(u_0,\ldots,u_L)=
\Psi_{\mathbf a}(u_0,\ldots,u_m+\pi\tau,\ldots,u_L)
e^{i\left((a_{m-1}\eta+\zeta)(a_m-a_{m-1})+(L-1)u_m-\sum_{i\ne m} u_i\right)}
$ for some $m=1,\ldots,L$.
\end{enumerate}
Furthermore, $\Psi$ and $\Psi'$ satisfy \eqref{qkzeq} with the same value of $\kappa$.
\end{lem}
\begin{proof}
The statement (1) is trivial 
because
the weights $W$ of \eqref{eq:weights} are invariant under $u\mapsto u+\pi$ and $\zeta\mapsto\zeta+\pi$.

We consider the next statement (2), i.e.,
$\zeta\mapsto \zeta+\pi\tau$.
We have, writing explicitly the dependence on $\zeta$, the following
``gauge'' transformation
\begin{align*}
  \Wop[a][b][c][d][u-u';\zeta+\pi\tau] = \Wop[a][b][c][d][u-u';\zeta] e^{i(c-b)} \frac{h(a,c|u)h(c,d|u')}{h(b,d|u)h(a,b|u')}
\end{align*}
with $h(a,c|u)=e^{i(c-a)u}$.

After some compensations in the $q$KZ equation \eqref{qkzeqcoeff}, it is not hard to see that 
$\sum_{\mathbf a}\ket{a}\Psi_{\mathbf a}(u_0+\pi\tau,\ldots) e^{3i(a_L\eta+\zeta)+i\sum_{j=1}^L (a_j\eta+\zeta)} \prod_{j=1}^L h(a_{j-1},a_j|u_j)$
satisfies the same $q$KZ equation at level $1$ as $\Psi$. This coincides once again with 
the expression in the proposition.

Finally, for statement (3), i.e.,
$u_m\mapsto u_m+\pi\tau$ with $m=1,\ldots,L$,
we use the gauge transformation
\begin{align*}
  \Wop[a][b][c][d][u+\pi\tau] = e^{-3i\eta}\,\Wop[a][b][c][d][u] \frac{g(a,c)}{g(b,d)}=e^{-3i\eta}\,\Wop[a][b][c][d][u]\frac{g(a,b)}{g(c,d)}
\end{align*}
with $g(a,c)=e^{i (a\eta+\zeta)(c-a)}$.
We then consider Eq.~\eqref{qkzeqcoeff} and apply the first form of the gauge transformation for $i=m$,
and its second form for $i\ne m$. In both cases we find that the vector with
entries $g(a_{m-1},a_m)\Psi_{\mathbf a}(\ldots,u_m+\pi\tau,\ldots)$
satisfies Eq.~\eqref{qkzeqcoeff} with same $\ell$ 
but modified constant $\kappa'=\kappa e^{-3i\eta}$ in the first case, 
$\kappa'=\kappa e^{3i(L-1)\eta}$ in the second.
This constant mismatch can be fixed for $\ell=1$ by considering
$
g(a_{m-1},a_m)
\,
e^{i((L-1)u_m-\sum_{i\ne m} u_i)}
\,
\Psi_{\mathbf a}(\ldots,u_m+\pi\tau,\ldots)
$
which coincides with the expression in the proposition.
\end{proof}

Considering that we have two independent solutions, it is natural
to suspect that they are related to each other by such translations.
Indeed, 
define the lattice $\Lambda$ inside $\mathbb C^{L+1}\ni (u_0,u_1,\ldots,u_L)$:
\[
\Lambda:=
(\pi \mathbb Z+ \pi\tau \mathbb Z)^{L+1}
\]
and its sublattice of index 2:
\[
\Lambda^e=\left\{
(\pi(m_0+n_0\tau),
\ldots,\pi(m_L+n_L\tau))\in\Lambda:
\sum_{i=0}^L n_i = 0\pmod 2\right\}
\]

By direct computation from the integral expressions
\ref{def:Psi}, one obtains:
\begin{prop}\label{prop:pseudoperiod}
$\Psi^{(k)}_{\mathbf a}$ has pseudo-periodicity lattice $\Lambda^e$, $k=2,3$;
$\Psi^{(2)}_{\mathbf a}$ and $\Psi^{(3)}_{\mathbf a}$ are exchanged by odd elements of $\Lambda$,
namely
\begin{align*}
\Psi^{(k)}_{\mathbf a}(\ldots,u_i+\pi,\ldots)&=(-1)^k
\Psi^{(k)}_{\mathbf a}(\ldots)
\\
\Psi^{(k)}_{\mathbf a}(u_0+\pi\tau,\ldots)
&= c_0\,e^{-i\left(3(a_L\eta+\zeta)+\sum_{j=1}^L (a_j\eta+\zeta)+\sum_{j=1}^L (a_j-a_{j-1})u_j\right)}
\Psi^{(\sigma(k))}_{\mathbf a}(\ldots)
\\
\Psi^{(k)}_{\mathbf a}(\ldots,u_m+\pi\tau,\ldots)
&=c_m\,
e^{-i\left((a_{m-1}\eta+\zeta)(a_m-a_{m-1})+(L-1)u_m-\sum_{i\ne m} u_i\right)}
\Psi^{(\sigma(k))}_{\mathbf a}(\ldots)
&
m=1,\ldots,L
\end{align*}
where $\sigma(k)=5-k$, and
the $c_m$ are certain constants (which can be calculated explicitly
but will not be needed).
\end{prop}

\subsection{Holomorphicity}
Another key property of $\Psi^{(j)}(\bu)$ 
is given by the following:
\begin{prop}\label{prop:holo}
The functions $\Psi^{(j)}(\bu)$ $(j=2,3)$ are holomorphic in all the arguments $u_m$ ($m=0,1,\ldots,L$).
\end{prop} 
\begin{proof}
From its definition \ref{def:Psi}, $\Psi^{(j)}$, as the integral of an expression
that is holomorphic in $\zeta$, is holomorphic in $\zeta=u_0$.

We consider next $m=1,\ldots,L$.
The first step in the proof is to use Prop~\ref{prop:pseudoperiod},
so that to prove holomorphicity, we need to demonstrate the absence of poles only within the fundamental domain $D_{\tau}$ of each of the 
arguments $u_m$.
Potential poles come from points where the arguments (or periodically shifted version of them) 
of the denominator of $\psi^{(j)}_{\ba,\bi}(u)$ vanish within this domain. These possibilities occur at
\begin{align*}
  u_m=u_{i_l},  \; 1\leq m\leq \alpha_\ell, \; m\neq i_\ell ;\quad\quad u_m=u_{i_l}+\eta,\;  \alpha_\ell\leq m \leq L.
\end{align*}
 We see immediately, due to the zero coming from the numerator term $\sli_{1\leq i<j\leq L} f(\eta-u_i+u_j)$ of $\psi^{(j)}_{\ba,\bi}(u)$, that
\begin{align*}
  \mbox{Res}\Big|_{ u_m=u_{i_l}+\eta} \psi^{(j)}_{\ba,\bi}(\bu) & = 0.
\end{align*}

\noindent For the remaining $u_m=u_{i_\ell}$ $(1\leq m\leq \alpha_\ell, m\neq i_\ell)$ possibility, we now show that the residue is zero.
There are two such cases to consider:\\
\begin{enumerate}
\item $m=i_{\ell'}$ for some  $\ell'\in\{1,2,\cdots,n\}/\{\ell\}$: In this case the residue is cancelled by the residue coming
from $\psi^{(j)}_{\ba,\bi'}(u)$ where $\bi'$ is the same as $\bi$ except that the values of $i_{\ell}$ and  $i_{\ell'}$ are exchanged.  

\item $m\neq i_{\ell'}$ for any $\ell'\in\{1,2,\cdots,n\}/\{\ell\}$: In this case the residue is cancelled by the residue coming
from $\psi^{(j)}_{\ba,\bi''}(u)$ where $\bi''$ is the same as $\bi$ except that the values of $i_{\ell}$ is replaced by $m$.  
\end{enumerate}

\noindent Thus all the non-zero residues of $\psi^{(0)}_{\ba,{\bf i}}(\bu)$ appearing $\Psi_{\ba}(\bu)$ are cancelled out which proves the proposition. To clarify the proof, consider the example considered above with $L=4$, $\ba=(3,4,5,4)$,  $\alpha=(2,3)$, and
$ S_{\ba}=\{(1,2),(1,3),(2,1),(2,3)\}$. We find for example
\begin{align*}
   \hbox{Res}_{u_1=u_2} \psi^{(j)}_{\ba,(1,2)}(\bu) & =  -\hbox{Res}_{u_1=u_2} \psi^{(j)}_{\ba,(2,1)}(\bu)\\
\hbox{Res}_{u_2=u_3} \psi^{(j)}_{\ba,(1,2)}(\bu) & =  -\hbox{Res}_{u_2=u_3} \psi^{(j)}_{\ba,(1,3)}(\bu).
\end{align*}
\end{proof}

\subsection{Theta functions}
Define the matrix $\alpha_{ij}$ depending on a sequence
$\mathbf a=(a_1,\ldots,a_L)$:
\[
\alpha_{ij}(\mathbf a)=
\begin{cases}
L+3& 
i=j=0
\\
a_j-a_{j-1}& 
i=0,\ j=1,\ldots,L
\\
a_i-a_{i-1}& 
j=0,\ i=1,\ldots,L
\\
L-1&
i=j=1,\ldots,L
\\
-1&
i\ne j,\ i,j=1,\ldots,L
\\
\end{cases}
\]
Combining Props.~\ref{prop:pseudoperiod} and \ref{prop:holo}, we conclude the following:
\begin{cor}\label{cor:theta}
$\Psi^{(k)}_{\mathbf a}$ ($k=2,3$) is a {\em theta function}\/ in the variables
$u_0=\zeta$ and $u_1,\ldots,u_L$, with prescribed lattice of pseudo-periods
$\Lambda^e$ and degree matrix 
given by $\frac{1}{2}\alpha_{ij}(\mathbf a)$, i.e.,
\[
\Psi^{(k)}_{\mathbf a}(\mathbf{u+\boldsymbol\lambda})
=
c^{(k)}_{\mathbf{\boldsymbol\lambda,a}}
e^{-i \sum_{i,j=0}^L n_i \alpha_{ij}(\mathbf a) u_j}
\Psi^{(k)}_{\mathbf a}(\mathbf u),
\qquad
\boldsymbol{\lambda}=(\pi(m_0+n_0\tau),
\ldots,\pi(m_L+n_L\tau))\in\Lambda^e
\]
where the $c_{\mathbf{\boldsymbol\lambda,a}}^{(k)}$ are some constants.
\end{cor}
In general, one expects solutions of the elliptic $q$KZ equation
to have complicated analytic behaviour (see the discussion in section~\ref{sec:disc}).
Here, the particular solutions which we exhibit have the simplest possibly
behaviour: they are holomorphic in all variables. Theta functions
are the natural analogue of polynomial functions in the elliptic world,
and in fact these solutions degenerate to polynomial solutions
in the trigonometric or rational limit, see App.~\ref{app:lim}.

Note that the
dimension of the space of such theta functions
is $2\det(\frac{1}{2}\alpha_{ij}(\mathbf a))_{i,j=0,\ldots,L-1}$ 
(the factor of $2$ being
the index of $\Lambda^e$ inside $\Lambda$),
which is given by
$d_n=(n+1)n^{2(n-1)}=2,12,324,20480\ldots$ for $L=2n=2,4,6,8,\ldots$

\subsection{Recurrence relations}
We now want to consider the relationship between solutions to the qKZB system of different sizes. We denote $\Psi^{(j)}(\bu)$ for system size $L$ by $\Psi^{(L,j)}(\bu)$. This notation is also used for other functions. In the following lemma we show that one can extract the state $\Psi^{(L,j)}(\bu)$ from the state $\Psi^{(L+2,j)}(\bu)$.
\begin{prop} \label{prop:rec}
For $\ket{\ba}\in\mathcal{B}_{L}$ and $j=2,3$ the following recursion relations hold:
\begin{align*}
 & (i) && \left. \Psi^{(L,j)}_{\ba}(\bu) \right|_{u_{L-1}=u_{L}+\eta} \\ 
 &&& = \delta_{a_{L}}^{a_{L-2}} (a_{L}-a_{L-1}) (-1)^{n} f(a_{L-1}\eta+\zeta) \prod_{k=1}^{L-2} f(2\eta-u_{k}+u_{L}) \Psi^{(L-2,j)}_{a_{1},\dots,a_{L-2}}(u_{1},\dots,u_{L-2}), \\
 & (ii) && \Psi^{(L,j)}_{\ba}(\bu) \\
 &&& = \frac{(-1)^{n} f(\eta)\left\{\Psi^{(L+2)}_{a_{1}\dots,a_{L},a_{L}+1,a_{L}}(u_{1},\dots,u_{L},u,u+\eta) - \Psi^{(L+2)}_{a_{1}\dots,a_{L},a_{L}-1,a_{L}}(u_{1},\dots,u_{L},u,u+\eta) \right\} }{f(2\eta)f(a_{L}\eta+\zeta)\prod_{k=1}^{L} f(2\eta-u_{k}+u)} 	
\end{align*}
\end{prop}
\begin{proof}
To prove the first relation we need to split the proof into $a_{L-1}=a_{L}+1$ and $a_{L-1}=a_{L}-1$. We start with the former which also necessitates that $\bar{\alpha}_{n}=L$ and $i_{n}=L$ for any $\bi\in \bS_{\bar{\alpha}}$. We first observe that in such a case
\begin{align*}
  & \bar{\psi}^{(j)}_{\ba,\bi}(\bu) \\
  & = \frac{\f[a_{L}\eta+\eta+\zeta]}{\f[\eta]} \left\{\prod_{1 \leq l < n} \f[u_{i_{l}}-u_{L}] \f[\eta-u_{i_{l}}+u_{L}] \right\} \prod_{l=1}^{n-1}  \f[a_{\bar{\alpha}_{l}}\eta+\eta+\zeta+u_{i_{l}}-u_{\bar{\alpha}_{l}}]  \\
  & \quad \times \fj[a_{L}\eta+\zeta -n\eta-2\sum_{l=1}^{n}u_{i_{l}} +\sum_{m=1}^{L}u_{m}] \left\{\frac{\prod_{1 \leq l < m \leq n-1} \f[u_{i_{l}}-u_{i_{m}}]}{\prod_{l=1}^{n-1}\prod_{\bar{\alpha}_{l} \leq m \leq L,\,m\neq i_{l}} \f[u_{i_{l}}-u_{m}]} \right\}  \\
  & \quad \times \left\{ \frac{\prod_{1 \leq i < j \leq L-1} \f[\eta-u_{i}+u_{j}] \prod_{1 \leq l < m \leq n-1} \f[\eta-u_{i_{l}}+u_{i_{m}}]}{\prod_{l=1}^{n-1} \prod_{1 \leq m \leq \bar{\alpha}_{l}} \f[\eta-u_{m}+u_{i_{l}}]} \right\}.
\end{align*}
This expression is highly dependent upon $i_{n-1}$
\begin{align*}
  & \left. \bar{\psi}^{(L,j)}_{\ba,\bi}(\bu) \right|_{u_{L-1}=u_{L}+\eta} \\
  & = \left\{ \begin{array}{ll}
    0 & : i_{n-1}=L-1 \\
    (-1)^{n-1} \frac{\f[a_{L-1}\eta+\zeta]}{\f[\eta]} \prod_{1 \leq i \leq L-2} \f[2\eta-u_{i}+u_{L}] \bar{\psi}^{(L-2,j)}_{a_{1},\dots,a_{L-2},i_{1},\dots,i_{n-1}}(u_{1},\dots,u_{L-2}) & : i_{n-1}<L-1
  \end{array} \right.
\end{align*}
This result, along with Eq. \mref{eq:intexp}, allow us to prove the recursion relation in this case:
\begin{align*}
  & \left. \Psi_{\tilde{a}}(\tilde{u}) \right|_{u_{L-1}=u_{L}+\eta} \\
  & = \delta_{a_{L-2}}^{a_{L}}\left[f(\eta)\right]^{n} f(a_{L}\eta+\zeta) \sum_{\tilde{i}\in \bar{I}^{(\eta)}_{\tilde{a}}, \, i_{n-1}<L-1} \left. \bar{\psi}^{(\eta)}_{\tilde{a},\tilde{i}}(\tilde{u})  \right|_{u_{L-1}=u_{L}+\eta} \\
  & = \delta_{a_{L-2}}^{a_{L}} (-1)^{n-1} f(a_{L-1}\eta+\zeta) \prod_{1 \leq i \leq L-2} f(2\eta-u_{i}+u_{L}) \Psi^{(L-2)}_{a_{1},\dots,a_{L-2}}(u_{1},\dots,u_{L-2}),
\end{align*}
as required. The $a_{L-1}=a_{L}-1$ follows similarly, however one should use the expression
\begin{align*}
  \Psi_{\ba}(\bu) 
  & = (-1)^{n} \left[\f[\eta]\right]^{n} \f[a_{L}\eta+\zeta] \sum_{\tilde{i}\in \bS_{\alpha}} \fj[a_{L}\eta+\zeta +n\eta+2\sum_{l=1}^{n}u_{i_{l}}- \sum_{m=1}^{L}u_{m}] \\
  & \quad \times \prod_{l=1}^{n} \f[a_{\alpha_{l}}\eta-\eta+\zeta-u_{i_{l}}+u_{\alpha_{l}}] \left\{ \frac{\prod_{1 \leq l < m \leq n} \f[u_{i_{l}}-u_{i_{m}}]}{\prod_{l=1}^{n} \prod_{\alpha_{l} \leq m \leq L,\, m\neq i_{l}} \f[u_{i_{l}}-u_{m}]} \right\}\\
  & \quad \times \left\{  \frac{\prod_{1 \leq i < j \leq L} \f[\eta-u_{i}+u_{j}] \prod_{1 \leq l < m \leq n} \f[\eta-u_{i_{l}}+u_{i_{m}}]}{\prod_{l=1}^{n}\prod_{1 \leq m \leq \alpha_{l}} \f[\eta-u_{m}+u_{i_{l}}]} \right\}.
\end{align*}

The second recursion relations follows from the first and the qKZB exchange relation. Recall that Eq. \mref{qkz1} states
\begin{align*}
  \Psi^{(L,j)}_{\ba}(u_{1},\dots,u_{L-2},u_{L},u_{L-1}) & = \sum_{\ket{\ba'}\in\mathcal{B}_{L}} \left[\prod_{k\neq L-1} \delta_{a_{k}}^{a_{k}'}\right] \Wop[a_{L-2}][a_{L-1}][a_{L-1}'][a_{L}][u_{L-1}-u_{L}]  \Psi^{(L,j)}_{\ba'}(\bu).
\end{align*}
To prove the second recursion relation one must first set $u=u_{L-1}=u_{L}-\eta$ and $a_{L-2}=a_{L}$. Next the first recursion relation must be applied to the left hand side while expanding the right hand side. Lastly, one replaces $L$ with $L+2$.
\end{proof}
One can use the rotation relation to generate a family of recursion relations at the specialisation
$u_{i+1}=u_i-\eta$.

\subsection{Wheel condition}
\begin{prop} \label{prop:wheel}
For $\Psi(\bu)=\Psi^{(2)}(\bu),\Psi^{(3)}(\bu)$, if $u_{j}=u_{i}-\eta$ and $u_{k}=u_{j}-\eta$ for some $1 \leq i < j < k \leq L$ then
\begin{align}
  \Psi(\mathbf{u}) & = 0 \label{eqnWheel}
\end{align}
\end{prop}
\begin{proof}
We start by seeing this lemma only applies with $L\geq 4$. We first consider the case $u_{L}=u_{L-1}-\eta=u_{L-2}-2\eta$. Using the recursion relations we find
\begin{align*}
  \Psi^{(L)}_{\ba}(\bu)
  & = \delta_{a_{L}}^{a_{L-2}} (a_{L}-a_{L-1}) (-1)^{n} f(a_{L-1}\eta+\zeta) \prod_{k=1}^{L-2} f(2\eta-u_{k}+u_{L}) \Psi^{(L-2)}_{a_{1},\dots,a_{L-2}}(u_{1},\dots,u_{L-2}) \\
  & = 0,
\end{align*}
due to the appearance of the term $f(2\eta-u_{L-2}+u_{L})$. Application of the exchange relation, Equation (\ref{qkz1}), proves the lemma. The reader should keep in mind that $R_{i}(\eta)$ is not defined and this leads to the restriction on the ordering of $i,j,k$.
\end{proof}

It would be interesting to study the space of theta functions with 
the same lattice of pseudo-periods $\Lambda^e$ and same degree
as $\Psi^{(k)}$
(cf~Cor.~\ref{cor:theta}) which satisfy the wheel condition \eqref{eqnWheel}.
The corresponding problem in the trigonometric case was solved by Kasatani
in \cite{Kasa-wheel}, leading to a a connection to nonsymmetric Macdonald
polynomials \cite{KT-qKZ}.

\section{The Combinatorial Line}\label{sec:combi}
The combinatorial line in the elliptic 8-vertex model occurs when $\eta=\frac{2\pi}{3}$. In this section we will see that when $\eta=\frac{2\pi}{3}$, which will be assumed for the entirety of the section, that the solutions of the qKZB system as defined in the previous section are eigenstates of the face transfer matrix. Depending on the value of the parameter $\zeta$ we find that the face model can be associated with either the 2-height RSOS model or the three colour problem.

\subsection{An Eigenstate of the Transfer Matrix} \label{secEigenstates}
On the combinatorial line we find that 
\begin{align*}
  S_{i} & = t(u_{i}), \\
  S_{i}\Psi^{(j)}({\bf u}) & = -\Psi^{(j)}({\bf u}),
\end{align*}
for $j=2,3$. This is the first hint that $\Psi^{(2)}({\bf u}),\Psi^{(3)}({\bf u})$ are eigenstates of the transfer matrix. To prove this indeed the case we will use the recursion relations determine in the previous section, however, we first require a useful identity. 

Consider a periodic chain with $L+2$ sites with the operator $R_{L}(u_{L+2}-u_{L}) \cdots R_{1}(u_{L+2}-u_{1}) \rho$ acting on either solution to qKZB system. The action of this operator can be computed either via the exchange relation, Eq. \mref{qkz1}, or the definition of the $R$-matrix, Eq. \mref{eq:Rmat}. Equating the two different ways as well as setting $u_{L+1}=u+\eta$ and $u_{L+2}=u$ leads to the identity
\begin{align*}
  & \Psi^{(L+2,j)}_{{\bf a}'}(u_{1},\dots,u_{L},u,u+\eta) \\
  & = -\sum_{\ket{{\bf a}}\in\mathcal{B}_{L+2}} \delta_{a_{L}}^{a_{L+1}'} \delta_{a_{L+1}}^{a_{L+2}'} \Wop[a_{L-1}'][a_{L}'][a_{L-1}][a_{L}][u-u_{L}] \cdots \Wop[a_{1}'][a_{2}'][a_{1}][a_{2}][u-u_{2}]  \\
  & \quad \times \Wop[a_{L+1}][a_{1}'][a_{L+2}][a_{1}][u-u_{1}] \Psi^{(L+2,j)}_{{\bf a}}(u_{1},\dots,u_{L},u+\eta,u),
\end{align*}
for $\ket{{\bf a}'}\in\mathcal{B}_{L+2}$ and $j=2,3$.

\begin{thm} \label{thmEigenstate}
On the combinatorial line $t(u)\Psi^{(j)}({\bf u}) = -\Psi^{(j)}({\bf u})$ for $j=2,3$.
\end{thm}
\begin{proof}
The approach take is to show that $\Psi^{(L,j)}({\bf u})$ being an eigenstate of $t(u)$ follows from $\Psi^{(L+2,j)}({\bf u})$ satisfying the qKZB system relations by utilising the recursion relations to relate the different states.
\begin{align*}
  & t(u) \Psi^{(L,j)}({\bf u}) \\
  & = \sum_{{\bf a},{\bf a}'\in\mathcal{B}_{L}} \Wop[a_{L-1}'][a_{L}'][a_{L-1}][a_{L}][u-u_{L}] \cdots \Wop[a_{1}'][a_{2}'][a_{1}][a_{2}][u-u_{2}] \Wop[a_{L}'][a_{1}'][a_{L}][a_{1}][u-u_{1}] \Psi^{(L,j)}_{{\bf a}}({\bf u}) \ket{{\bf a}'} \\
  & = \sum_{{\bf a},{\bf a}'\in\mathcal{B}_{L+2}} \frac{(-1)^{n+1}(a_{L+1}'-a_{L}')}{\f[a_{L}'\eta+\zeta]\prod_{j=1}^{L} \f[\eta-u+u_{j}]} \delta_{a_{L}'}^{a_{L+2}'} \delta_{a_{L}}^{a_{L+1}'} \delta_{a_{L+1}}^{a_{L+2}'} \Wop[a_{L+1}][a_{1}'][a_{L+2}][a_{1}][u-u_{1}]  \\
  & \quad \times \Wop[a_{L-1}'][a_{L}'][a_{L-1}][a_{L}][u-u_{L}] \cdots \Wop[a_{1}'][a_{2}'][a_{1}][a_{2}][u-u_{2}] \Psi^{(L+2,j)}_{{\bf a}}(u_{1},\dots,u_{L},u+\eta,u) \ket{a_{1}'\cdots a_{L}'} \\
  & = - \sum_{{\bf a}'\in\mathcal{B}_{L+2}} \frac{(-1)^{n+1}(a_{L+1}'-a_{L}')}{\f[a_{L}'\eta+\zeta]\prod_{j=1}^{L} \f[\eta-u+u_{j}]} \delta_{a_{L}'}^{a_{L+2}'} \Psi^{(L+2,j)}_{{\bf a}'}(u_{1},\dots,u_{L},u,u+\eta) \ket{a_{1}'\cdots a_{L}'} \\
  & = - \Psi^{(L,j)}({\bf u}) 
\end{align*}
for $j=2,3$.
\end{proof}
On the combinatorial line we also find that due to the periodicity of $\f[u]$, the face weights are unchanged when all heights are shifted by three. This leads to a symmetry of the transfer matrix:
\begin{align*}
  t(u)T & = Tt(u) \quad \mbox{where} \quad \bra{{\bf a}'} T \ket{{\bf a}} =\prod_{i=1}^{L} \delta_{a_{i}+3}^{a_{i}'},
\end{align*}
for $\ket{{\bf a}},\ket{{\bf a}'}\in\mathcal{B}_{L}$. The operator $T$ is also a symmetry of the $R$-matrix. From Defs.~\ref{def:Idef} and \ref{def:Psi}, $\Psi^{(2)}({\bf u}),\Psi^{(3)}({\bf u})$
are eigenstates with eigenvalue one:
\begin{align*}
  T\Psi^{(j)}({\bf u}) & = \Psi^{(j)}({\bf u}),\qquad j=2,3.
\end{align*}
These properties allow for the construction of related integrable models with finite-dimensional Hilbert spaces. The models depend upon the parameter $\zeta$ and for each the state $\Psi({\bf u})$ maps to a state which satisfies the qKZB equations and is an eigenstate of the transfer matrix.

\subsection{The Three Colour Problem}
In this section we define the three colour problem for $\zeta\notin\frac{\pi}{3}\Z$ and show that $\Psi^{(2)}({\bf u}),\Psi^{(3)}({\bf u})$ can be mapped onto solutions to the qKZB system of the new model.

The Hilbert space of periodic three colour problem on a chain of $L$ sites, $\mathcal{H}_{L}^{3C}$, is defined as the complex span of the following set of basis vectors
\begin{align*}
  \mathcal{B}_{L}^{3C} & = \left\{\ket{\bar a_{1},\bar a_{2},\ldots,\bar a_{L}}\left|\right. \bar a_{i}\in \Z/3\Z, \, \bar a_{i-1}-\bar a_{i}=\pm 1 \right\}.
\end{align*}
The Hilbert space has dimension $2^{L}+2(-1)^{L}$, with odd length chains allowed. The face weights of the three colour problem are defined as
\begin{align*}
  \Wopc[\bar a][\bar b][\bar c][\bar d][u] & = \Wop[a][b][c][d][u] \quad \mbox{where} \quad \bar x = x \mod 3,
\end{align*}
and $a,b,c,d\in\Z$, $|a-b|=|a-c|=|b-d|=|c-d|=1$. All other face weights are defined to be zero. Using these face weights one can define $R$-matrices, $R^{3C}_{i}(u)$ and a transfer matrix, $t^{3C}(u)$, using Equations (\ref{eq:Rmat}) and (\ref{eq:tmat}), respectively, with weights $W$ replaced by $W^{3C}$.

Each sequence $\bar a: \Z/L\Z \to \Z/3\Z$ in $\mathcal{B}_L^{3C}$ can be lifted to a sequence $a: \Z/L\Z \to \Z$ with $a_{i+1}-a_i=\pm 1$; the latter is defined uniquely up to overall shift by $3$. This allows to divide $\mathcal{B}_L^{3C}$ into sectors:
\[
\mathcal{B}_L^{3C} = \bigsqcup
_{\substack{|k|\le L \\ k=L\pmod 2 \\ k=0\pmod 3}}
\mathcal{B}_{L;k}^{3C},
\qquad
\mathcal{B}_{L;k}^{3C}=
\left\{ \ket{\bar a_1,\bar a_2,\ldots,\bar a_L}: a_{i+L}=a_i+k\right\}
\]
The corresponding decomposition of $\mathcal{H}_L^{3C}=\bigoplus_k \mathcal{H}_{L;k}^{3C}$ is preserved by the dynamics. In particular, when $L$ is even,
$\mathcal{H}_{L;0}^{3C}$ naturally embeds itself inside $\mathcal{H}_L$:
\[
\phi \ket{\bar a_1,\ldots,\bar a_L} = \sum_{i=-\infty}^{+\infty} \ket{a_1+3i,a_2+3i,\ldots,a_L+3i}
\]
$\phi$ induces an isomorphism from $\mathcal{H}_{L;0}^{3C}$ to its
image, which is precisely 
$\{v\in\mathcal{H}_{L}| Tv=v \}$.
In what follows, we restrict $R_i^{3C}(u)$ and $t^{3C}(u)$ to the sector
$\mathcal{H}_{L;0}^{3C}$.

The map $\phi$ has the following properties:
\begin{align*}
  \phi R_{i}^{3C}(u) & = R_{i}(u) \phi \\
  \phi t^{3C}(u) & =  t(u)\phi
\end{align*}
As $\Psi^{(2)}({\bf u}),\Psi^{(3)}({\bf u})$ are invariant under the action of $T$ it follows that the states $\phi^{-1}\Psi^{(2)}({\bf u}),\phi^{-1}\Psi^{(3)}({\bf u})$ must be solutions of the qKZB system and eigenstates of the transfer matrix, $t^{3C}(u)$.

In other models (see \cite{artic31,artic42,artic59}), 
the solution to the qKZ(B) system
at the special point where the shift $s$ vanishes, is not only an eigenstate 
but also the ground state of the model. This, however, is not the case for the three colour problem. If we instead interpret the model as a quantum chain with Hamiltonian
\begin{align}
  H & = \left. \frac{d}{du} \ln(t^{3C}(u)) \right|_{u=0}. \label{eq:Ham}
\end{align}
For small system sizes ($L=4,6$) we find that in the homogeneous case, $u_{1}=\cdots =u_{L}=0$, $\phi^{-1}\Psi^{(2)}({\bf u}),\phi^{-1}\Psi^{(3)}({\bf u})$ are the first excited states. We conjecture this to be true for larger system sizes. If we take the trigonometric limit, $\tau\rightarrow 0$, then we find that $\phi^{-1}\Psi^{(2)}({\bf u}),\phi^{-1}\Psi^{(3)}({\bf u})$ correspond to ground states of twisted sectors of the model (see Appendix \ref{app:lim}). We suspect that there are analogous sectors in the elliptic case and determining these would provide a stepping stone to finding suitable boundary conditions for obtaining solutions to the qKZB system in the vertex language.

\subsection{The 2-Height RSOS model}

In the following we show that when $\zeta \in \frac{\pi}{3}\Z$ the model reproduces the 2-height RSOS model. Furthermore, the solutions to qKZB system of the SOS model are mapped to solutions to the qKZB system of the RSOS model and are both eigenstates of the transfer matrix. We show this for $\zeta=0$, assumed for all the computations below, and discuss the generalisation.

We consider the periodic RSOS model on a chain of $L=2n$ sites with the Hilbert space $\mathcal{H}_{L}^{RSOS}$ defined as the complex span of the following set of basis vectors
\begin{align*}
  \mathcal{B}_{L}^{RSOS} & = \left\{\ket{a_{1},a_{2},\dots a_{L}}\left|\right. a_{i}\in \{1,2\}, \, |a_{i-1}-a_{i}|=1 \right\}.
\end{align*}
It is clear that the Hilbert space has dimension 2 regardless of chain length and is a subspace of both $\mathcal{H}_{L}$ and $\mathcal{H}_{L}^{3C}$. Despite $R_{i}(u)$ and $t(u)$ no longer being defined, due to the existence of poles, on all of $\mathcal{H}_{L}$ both are defined $\mathcal{H}_{L}^{RSOS}$. Furthermore,  we find that $\mathcal{H}_{L}^{RSOS}$ is invariant under the action of $R_{i}(u)$ and $t(u)$. This implies one can construct an integrable model associated to the space $\mathcal{H}_{L}^{RSOS}$, this model is precisely the standard 2-height elliptic RSOS model \cite{ABF}.

Due the homomorphicity of the two solutions to the qKZB equation both are still valid. However, one needs to consider the effect of setting $\zeta=0$ on coefficients $\Psi_{{\bf a}}({\bf u})$. We first compute that if $i_{1}=1$ and $\alpha_{1}=1$ then
\begin{align*}
  \psi^{(j)}_{{\bf a},{\bf i}}({\bf u})
  & = \fj[a_{L}\eta -n\eta+u_{1}+2\sum_{l=2}^{n}u_{i_{l}}- \sum_{m=2}^{L}u_{m}] \f[a_{1}\eta] \prod_{l=2}^{n} \f[a_{\alpha_{l}}\eta-u_{i_{l}}+u_{\alpha_{l}}] \\
  & \quad \times  \left\{ \frac{\prod_{1 < m \leq n} \f[u_{1}-u_{i_{m}}]}{\prod_{2 \leq m \leq \alpha_{1}} \f[u_{m}-u_{1}]} \right\} \times \left\{ 	\frac{\prod_{1 \leq i < j \leq  L} \f[\eta-u_{i}+u_{j}]\prod_{1 < m \leq n} \f[\eta-u_{1}+u_{i_{m}}]}{ \prod_{\alpha_{1} \leq m \leq  L} \f[\eta-u_{1}+u_{m}]} \right\} \\
  & \quad \times \left\{ \frac{\prod_{2 \leq l < m \leq n} \f[u_{i_{l}}-u_{i_{m}}]}{\prod_{l=2}^{n}\prod_{1 \leq m \leq \alpha_{l},\, m\neq i_{l}} \f[u_{m}-u_{i_{l}}]} \right\} \times \left\{ \frac{\prod_{2 \leq l < m \leq n} \f[\eta-u_{i_{l}}+u_{i_{m}}]}{\prod_{l=2}^{n} \prod_{\alpha_{l} \leq m \leq  L} \f[\eta-u_{i_{l}}+u_{m}]} \right\}.
\end{align*}
This allows us to see that $\Psi_{{\bf a}}({\bf u})$ must have factors $\f[a_{L}\eta]$ and $\f[a_{1}\eta]$. The rotation condition implies that $\Psi^{(j)}_{{\bf a}}({\bf u})=0$ whenever any height appearing in ${\bf a}$, say $a_{i}$, is a multiple of 3. Consequently $\phi^{-1}\Psi^{(j)}$ belongs to the subspace $\mathcal{H}_{L}^{RSOS}$, satisfies the qKZB equation and is an eigenstate of the transfer matrix.
In terms of the RSOS basis the transfer matrix is
\begin{align*}
  t(u) & = \left(\begin{array}{cc} 0 & 1 \\ 1 & 0 \end{array}\right).
\end{align*}
As the two solutions to the qKZB system we constructed have the same eigenvalue we conclude that $\Psi^{(2)}({\bf u})=\Psi^{(3)}({\bf u})$, in contrast to the original SOS model and three colour problem. This has been verified for small system sizes.

This approach can be generalised for any $\zeta \in \frac{\pi}{3}\Z$ as it only relies upon $\f[a\eta+\zeta]$ being zero for certain values of the height $a$. Naturally the heights allowed in the basis states may differ but the resulting model is equivalent.

\section{Discussion}\label{sec:disc}
To summarise: in this paper, we have constructed  solutions to the level-1 qKZB equation \eqref{qkzeq} as the integral expressions \eqref{eq:intform1} (see Theorem \ref{thm:qkz}).
These solutions are pseudo-periodic (Proposition \ref{prop:pseudoperiod}),  holomorphic (Proposition \ref{prop:holo}), and therefore theta functions (Corollary \ref{cor:theta}) in both dynamical and spectral parameters; they satisfy recurrence relations (Propositions \ref{prop:rec}) and the wheel condition (Proposition \ref{prop:wheel}). Along the combinatorial line, they correspond to eigenvectors of the SOS periodic transfer matrix with an even number of sites. We have discussed the connection with eigenvectors of the (R)SOS three-colour problem and 2-Height model when the free parameter is $\neq \frac{\pi}{3}\Z$ and $=\frac{\pi}{3}\Z$ respectively.  Based on the evidence of size $4$ and $6$ chains we conjecture that the our qKZB solutions correspond to the first excited states of the  three-colour problem in the
 homogeneous $u_1=u_2=\cdots = u_L$ case. 

Let us now discuss some connections with existing work. Firstly, solutions of qKZB equations are constructed in \cite{FTV-qKZB}. Our 
solutions appear to differ
from these solutions in several respects:  the integrands of our solutions  \eqref{eq:intform1} are theta functions, that is infinite products, whereas the integrands of those in  \cite{FTV-qKZB} are double infinite products.
Also, the solutions of  \cite{FTV-qKZB} are ``hypergeometric integrals'' in the sense that they generalise 
the generic hypergeometric solution of the KZ equation. In comparison, our solutions reduce to {\em polynomials}\/ in the trigonometric limit, and {\em a fortiori}, in the KZ limit.

A second way of constructing solutions of the  qKZB equations is provided by the vertex operator approach to SOS models described in \cite{Fodaetal94}. In this approach, qKZB solutions are given in terms of the trace of products of vertex operators, which can be written as
multiple integral expressions using the free-field realisation of \cite{LukPugai96}. 
However, the required trace of vertex operators ceases to be well-defined in the `critical-level' limit $\ell=\rightarrow -2$ -- which is precisely the point at which the trace would correspond to an eigenvector of the SOS transfer matrix. For the combinatorial-line case with $\eta=2\pi/3$ these trace solutions are also ill-defined at $\ell=1$. Our approach has been to construct new solutions, specifically defined at $\ell=1$ for generic $\mu$ and for which the combinatorial-line $\eta=2\pi/3$ limit is well-defined. 
We then use the result that along the combinatorial line these $\ell=1$ solutions correspond to eigenvectors of the transfer matrix;
this fact is ultimately due to the simple observation that along this line the shift parameter in the qKZB equation is $s=3\eta=2\pi$ and hence that the 
$R(u+s)$ which appears in the qKZB equation \eqref{qkzeq} is equal to $R(u)$.  This strategy is similar to that used for the trigonometric case in \cite{artic31}, with one key difference. In \cite{artic31}, the level $\ell=1$ qKZ solutions were themselves constructed by computing matrix elements (as opposed to traces) of vertex operators between level-1 highest weight states \cite{JM-book}. In the trigonometric case, such matrix elements obey level-1 qKZ equations as was shown in \cite{FR-qKZ}. In the elliptic case, analogous level-1 one matrix elements can also be computed but do {\it not} obey level-1 qKZB equations
. It is for this reason that we were led to seek the independent construction of level-1 solutions presented in this paper.

Finally, let us consider the connection with the earlier work \cite{artic57} on the 8-vertex on the 
combinatorial line. 
There are some technical differences with \cite{artic57}:
our system has {\em even}\/ size, which is necessary in a height model
with periodic boundary conditions where heights on neighbouring sites
differ by $\pm 1$. This prevents a direct identification
of the present model with the eight-vertex model of \cite{artic57}
via the vertex-IRF transformation \cite{Baxter1973}. Furthermore, in the current work, 
we are able to provide explicit formulae
for the entries of the solution of the qKZB system, leaving the inhomogeneities
free -- in \cite{artic57}, they were ``half-specialised'' to pairs of opposite values,
resulting in expressions involving only even theta functions and therefore a rational parameterisation;
here no such parameterisation is available.


\begin{appendix}
\section{The \texorpdfstring{$L=4$}{L=4} Case}\label{app:ex}
\def\th{\vartheta}
We describe here our solutions of the $q$KZB system in size $L=4$.
We only consider here one solution, $\Psi=\Psi^{(3)}$, noting that the other
one can be recovered by using Prop.~\ref{prop:pseudoperiod}.
Once one height is fixed, say $a_L=a$, $\Psi$ has $\binom{4}{2}=6$ components.

We start with the trivial component (Eq. \mref{qkzupstate})
\begin{multline*}
\Psi_{a+1,a+2,a+1,a}=
\th_1(a\eta+\zeta)
\th_1((a+1)\eta+\zeta)
\th_1((a+2)\eta+\zeta)
\\
\th_1(u_2-u_1+\eta)
\th_1(u_4-u_3+\eta)
\th_3((a-2)\eta+\zeta+u_1+u_2-u_3-u_4|2\tau)
\end{multline*}
(the quasi-period is $\tau$ unless otherwise specified).
Using the cyclicity relation \eqref{qkz2}, we can obtain 3 more components
which also have the same simple factorised form.

The 2 remaining components are obtained from each other by cyclic rotation.
We consider here $(a+1,a,a+1,a)$.
According to Cor.~\ref{cor:theta},
we have to search inside a space of theta functions of dimension $12$, 
but actually, by application of relation (iii) of Lem.~\ref{lem:qkz1}
(which allows to compute this component out of $\Psi_{a+1,a+2,a+1,a}$),
it has a prefactor of $\th_1(a\eta+\zeta)\th_1((a+1)\eta+\zeta)$, 
which reduces the search to a space of dimension $4$, spanned by:
\begin{multline*}
f_k=\th_1(a\eta+\zeta)\th_1((a+1)\eta+\zeta)\th_k((a+\frac{1}{2})\eta+\zeta)
\hfill (k=1,\ldots,4)
\\
\th_k(u_4-u_2+\frac{3}{2}\eta)\th_k(u_3-u_1+\frac{3}{2}\eta)\th_{4-\lceil k/2\rceil}(u_4+u_2-u_3-u_1+(a-1)\eta+\zeta|2\tau)
\end{multline*}
Then $\Psi_{a,a-1,a,a-1,a}=\sum_{k=1}^4 \alpha_k f_k$, with
\[
\alpha_k 
=(-1)^{\frac{(k-3)(k-4)}{2}}\frac{1}{\th'_1(0)} \prod_{\ell\ne k}\th_\ell(\eta/2)
\]

\section{The Trigonometric and Rational Cases} \label{app:lim}
\subsection{The Trigonometric Case}
The trigonometric SOS model is defined on the same Hilbert space as the elliptic SOS model with face weights again defined by Equation (\ref{eq:weights}) while setting $f(u)=\sin(u)$. The $R$-matrices and transfer matrix are defined by Equations (\ref{eq:Rmat}) and (\ref{eq:tmat}), respectively. If we consider the limit $\tau\rightarrow 0$, i.e. setting the elliptic nome to one then we find
\begin{align*}
  \lim_{\tau\rightarrow 0} \vartheta_{1}(u|\tau)\propto \sin(u), \hspace{1cm}
  \lim_{\tau\rightarrow 0} \vartheta_{2}(u|\tau)\propto \cos(u), \hspace{1cm}
  \lim_{\tau\rightarrow 0} \vartheta_{3}(u|\tau)\propto 1,
\end{align*}
where the proportionality is unimportant for the purpose of this discussion. From these relations we are able to view the trigonometric SOS model as a limit of the elliptic SOS model and infer results. Setting $f^{(2)}(u)=\cos(u)$ and $f^{(3)}(u)=1$ implies that the states $\Psi^{(2)}({\bf u})$ and $\Psi^{(3)}({\bf u})$ are solutions of the qKZ(B) system.

Solutions of the qKZ(B) system for the trigonometric and elliptic cases also share similar properties, with the properties of the latter discussed in detail in Section~\ref{sec:anal}. In the trigonometric case $\Psi^{(2)}({\bf u})$ and $\Psi^{(3)}({\bf u})$ both satisfy the recursion relations and wheel condition defined in Propositions \ref{prop:rec} and \ref{prop:wheel} . The coefficients of the solutions of the qKZ system are Laurent polynomials in the variables $q=e^{i\eta}$, $z_{j}=e^{iu_{j}}$ and $Z=e^{i\zeta}$. Moreover, they are homogeneous with respect to the variables $z_{1},\dots,z_{L}$ with degree zero. The most negative/positive degree of $\Psi^{(2)}_{{\bf a}}({\bf u})$ and $\Psi^{(3)}_{{\bf a}}({\bf u})$ in terms of a single variable $z_{j}$ is $n$ and $n-1$, respectively. The Laurent polynomial nature and degree of the solutions is analogous to the holomorphicity and pseudo-periodicity of the elliptic case.

The trigonometric limit of the elliptic SOS model also intersects the combinatorial line when $\eta=\frac{2\pi}{3}$ and is referred to as the combinatorial point. As such it follows that at this point $\Psi^{(2)}({\bf u}), \Psi^{(3)}({\bf u})$ are eigenstates of the transfer matrix. Furthermore, the maps to the three colour problem and 2-height RSOS model can still be applied, albeit to trigonometric versions.

It was mentioned that in the elliptic three colour problem the solutions of the qKZB system are the first excited state and not ground states. It was also stated that the authors suspect that solutions of the qKZB system are ground states of symmetry sector of the model. To explain this we use the quantum chain interpretation with the homogeneous Hamiltonian, given by Equation (\ref{eq:Ham}), of the trigonometric three-colour problem. The Hamiltonian can be written as
\begin{align*}
  H  & = \frac{2}{\sqrt{3}}\sum_{k} \left[ U_{k} + I \right]
\end{align*}
where
\begin{align*}
  U_{k} & = -\frac{\sqrt{3}}{2}\left.\frac{d}{du} \ln\left(R_{k}(u)\right) \right|_{u=0} - I, \\
  U_{k}U_{k\pm1}U_{k} & = U_{k}, \\
  U_{k}U_{k} & = 2\Delta U_{k}.
\end{align*}
with $\Delta = \cos(\eta) = -\frac{1}{2}$. This implies it is connected to the Temperley--Lieb model and subsequently the XXZ model with anisoptropy $\Delta=\cosh(\eta)$ \cite{AufKlu2010}. This allows us to use the Ansatz \cite{Nepome2003}
\begin{align*}
  \Lambda(u) & = \omega \left[\frac{\sin(u+\eta)}{\sin(u-\eta)}\right]^{L} \frac{q(u-\eta)}{q(u)} + \omega^{-1} \left[\frac{\sin(u)}{\sin(u-\eta)}\right]^{L} \frac{q(u+\eta)}{q(u)}
\end{align*}
where $\Lambda(u)$ is an eigenvalue of the homogeneous transfer matrix, $\omega\in\C$ is a twist and $q(u)=\prod_{i=1}^{l}\sin(u-u_{i}+\frac{\eta}{2})$. The Bethe equations are given by
\begin{align*} 
  \left(\frac{\sin\left(u_{j} - \frac{\eta}{2}\right)}{\sin\left(u_{j} + \frac{\eta}{2}\right)}\right)^{L} 
  & = - \omega^{2} \prod_{k=1}^{l} \left(\frac{\sin\left(u_{j}-u_{k} - \eta\right)}{\sin\left(u_{j}-u_{k} + \eta\right)}\right).
\end{align*}
The energy and momentum of the eigenstates can be given in terms of the Bethe roots, however, these are not of interest to us. The quantity of interest is the twist $\omega$, which can be expressed in terms of the eigenvalue of the homogeneous transfer matrix,
\begin{align*}
  \lim_{u\rightarrow\pm\infty} \Lambda(u) 
  & = \omega e^{\pm i\eta(\frac{L}{2}-l)} + \omega^{-1} e^{\pm i\eta(l-\frac{L}{2})}
\end{align*}
For system sizes $L=2,\dots,10$ we have found that the value $\lim_{u\rightarrow\pm\infty} \Lambda(u)$ is either $-1$ or $2$ when $L$ is even and either $1$ or $-2$ when $L$ is odd. This constrains the allowed twists and implies for even $L$ that the only twists allowed are cube roots of unity i.e. $\omega\in\{1,e^{i\eta},e^{-i\eta}\}$.

Consider only the case for $L$ even, we find for small system sizes ($L=2,\dots,8$) that there is one state with negative energy, two states with zero energy while every other state has positive energy. It was observed for $L=2,\dots,6$ that the ground state has negative energy, momentum $0$, twist $1$ and $l=\frac{L}{2}$ Bethe roots while the two solutions of the qKZ system of the three colour problem have energy $0$, momentum $\pi$, twist $e^{\pm i\eta}$ and $l=\frac{L}{2}$ Bethe roots. In fact Theorem \ref{thmEigenstate} implies that the solutions of the qKZ system must have energy $0$ and momentum $\pi$. Solving the Bethe equations numerically for larger $L$ yields results consistent with the observations stated.
Thus it appears that the true ground state of the system corresponds to the ground state of the spin-$1/2$ XXZ chain at the combinatorial point while the solutions of the qKZ states correspond to ground states of the twisted spin-$1/2$ XXZ chain at the combinatorial point.

\subsection{The Rational Case}
The rational SOS model is defined on the same Hilbert space as the trigonometric and elliptic SOS models with $f(u)=u$. Like the trigonometric case it can be considered a limiting case of the elliptic model, in fact it is also a limiting case of the trigonometric case. This relationship to trigonometric and elliptic models is seen via the relations
\begin{align*}
  \lim_{\tau,t\rightarrow 0} \vartheta_{1}(ut|\tau)\propto u, \hspace{1cm}
  \lim_{\tau,t\rightarrow 0} \vartheta_{2}(ut|\tau)\propto 1+u^{2}/2, \hspace{1cm}
  \lim_{\tau,t\rightarrow 0} \vartheta_{3}(ut|\tau)\propto 1,
\end{align*}
where again the proportionality is unimportant for the purpose of this discussion.

Since solutions of the qKZ system form a vector space we can set $f^{(2)}(u)=u^{2}$ and $f^{(3)}(u)=1$, which ensures that $\Psi^{(2)}({\bf u})$ and $\Psi^{(3)}({\bf u})$ are solutions of the qKZ system that satisfy recursion relations and the wheel condition. Additionally, each $\Psi^{(2)}_{{\bf a}}({\bf u})$ and $\Psi^{(3)}_{{\bf a}}({\bf u})$ is a homogeneous polynomial in the variables $\eta,\zeta,u_{1},\dots,u_{L}$ of degree $n^{2}+3$ and $n^{2}+1$, respectively.

The rational limit does not intersect with the combinatorial line. Thus the model can not be mapped to either a 2-height RSOS model or the three colour problem. This can also be seen as a consequence of $f(u)$ not being periodic.


\tikzset{bgplaq/.style={fill=gray!20!white}}

\def\wplaq(#1,#2,#3){
\begin{scope}[shift={(#1,#2)}]
\draw (0,0) rectangle ++(1,1); 
\draw (0,0.8) -- (0.2,1);
\draw (0.5,0.5) node {\small $ #3 $};
\end{scope}
}

\def\wsplaq(#1,#2,#3){
\begin{scope}[shift={(#1,#2)}]
\draw (0,0) rectangle ++(1,1); 
\draw (0,0.8) -- (0.2,1);
\draw (0.5,0.5) node {\tiny $#3$};
\end{scope}
}

\def\wlabels(#1,#2,#3,#4,#5,#6){
\begin{scope}[shift={(#1,#2)}]
\draw (0,1) node[above] {\small $#3$}; \draw (1,1) node[above] {\small $#4$};\draw (1,0) node[below] {\small $#5$};\draw (0,0) node[below] {\small $#6$};
\end{scope}
}

\def\qkzplaq(#1,#2,#3){
\begin{scope}[shift={(#1,#2)}]
\draw[bgplaq] (0,0) rectangle ++(1,1); 
\draw[blue,ultra thick] (0,0) -- (1,0);
\draw (0,0.8) -- (0.2,1);
\draw (0.5,0.5) node {\small $#3$};
\end{scope}
}

\def\downqkzplaq(#1,#2,#3){
\begin{scope}[shift={(#1,#2)}]
\draw[bgplaq] (0,0) -- (1,-0.5) -- (1,0.5) -- (0,1) -- (0,0);
\draw[blue,ultra thick] (0,0) -- (1,-0.5);
\draw (0,0.8) -- (0.2,0.9);
\draw (0.5,0.25) node {\small $#3$};
\end{scope}
}

\def\upqkzplaq(#1,#2,#3){
\begin{scope}[shift={(#1,#2)}]
\draw[bgplaq] (0,-0.5) -- (1,0) -- (1,1) -- (0,0.5) -- (0,-0.5);
\draw[blue,ultra thick] (0,-0.5) -- (1,0);
\draw (0,0.3) -- (0.2,0.6);
\draw (0.5,0.25) node {\small $#3$};
\end{scope}
}

\def\dplaq(#1,#2,#3){
\begin{scope}[shift={(#1,#2)}]
\draw (0,0) -- (1,0) -- (0,0.5) -- (-1,0.5) -- (0,0);
\draw (-0.8,0.4) -- (-0.8,0.5);
\draw (0,0.25) node {\tiny $ #3 $};
\end{scope}
}


\section{A Pictorial Representation of the qKZB System}\label{app:graph}
In this appendix, we introduce a simple graphical representation of the qKZB system. Firstly, we use the 
conventional graphical representation of the SOS weights:
\begin{align*}
  \Wop[a][b][c][d][u] & = \begin{tikzpicture}[baseline=10pt,scale=1]
  \wplaq (0,0,u); \wlabels(0,0,a,b,c,d);
  \end{tikzpicture}
\end{align*}
The transfer matrix defined in Equation  \eqref{eq:tmat} is then given by
\begin{align*}
  \bra{\ba'}t(u) \ket{\ba} & = 
  \begin{tikzpicture}[baseline=10pt,scale=1]
  \foreach \x in {1,...,4} 
  {
  \wplaq(\x,0,u_{\x});
  }
  \wplaq(5,0,\cdots);
  \wplaq(6,0,u_L);
  \wlabels (1,0,a'_L,a'_1,a_1,a_L); 
  \wlabels (3,0,a'_2,a'_3,a_3,a_2); 
  \wlabels (6,0,a'_{L-1},a'_L,a_{L},a_{L-1}); 
  \end{tikzpicture}
\end{align*}
Let us now represent a solution of the qKZB system \eqref{qkz1}-\eqref{qkz2} as follows:\\[2mm]
\begin{align*}
\Psi(u_1,u_2,\cdots,u_L) & = 
\begin{tikzpicture}[baseline=10pt,scale=1]
  \foreach \x in {1,...,4} 
  {
  \qkzplaq(\x,0,u_{\x});
  }
  \qkzplaq(5,0,\cdots);
  \qkzplaq(6,0,u_L);
  \end{tikzpicture}
\end{align*}
and components, with $\ba=(a_1,a_2,\cdots,a_L)$, by
\begin{align*}
  \Psi_{\ba}(u_1,u_2,\cdots,u_L) & = 
  \begin{tikzpicture}[baseline=10pt,scale=1]
  \foreach \x in {1,...,4} 
  {
  \qkzplaq(\x,0,u_{\x});
  }
  \qkzplaq(5,0,\cdots);
  \qkzplaq(6,0,u_L);
  \wlabels (1,0,a_L,a_1,,); 
  \wlabels (3,0,a_2,a_3,,); 
  \wlabels (6,0,a_{L-1},a_L,,); 
  \end{tikzpicture}
\end{align*}
The qKZB system (\ref{qkz1},\ref{qkz2}) can then be represented very simply as 
\begin{align*}
  \begin{tikzpicture}[baseline=10pt,scale=1]
  \qkzplaq(0,0,u_1); \qkzplaq(1,0,\cdots); \downqkzplaq(2,0,u_i); \upqkzplaq(3,0,u_{i+1});\qkzplaq(4,0,\cdots);\qkzplaq(5,0,u_L);
  \draw (2,1) -- (3,1.5) -- (4,1);
  \draw (2.2,0.9) -- (2.2,1.1);
  \draw (3,1) node {\tiny $u_i-u_{i+1}$};
  \end{tikzpicture}
  & =  \begin{tikzpicture}[baseline=10pt,scale=1]
  \qkzplaq(0,0,u_1); \qkzplaq(1,0,\cdots); \qkzplaq(2,0,u_{i+1}); \qkzplaq(3,0,u_{i});\qkzplaq(4,0,\cdots);\qkzplaq(5,0,u_L);
  \end{tikzpicture}
\end{align*}
and 
\begin{align*}
  \begin{tikzpicture}[baseline=10pt,scale=1]
  \qkzplaq(0,0,u_1); \qkzplaq(1,0,u_2); \qkzplaq(2,0,u_{3}); \qkzplaq(3,0,u_4);\qkzplaq(4,0,\cdots);\qkzplaq(5,0,u_L);
  \wlabels (0,0,a_L,a_1,,); 
  \wlabels (2,0,a_2,a_3,,); 
  \wlabels (5,0,a_{L-1},a_L,,); 
  \end{tikzpicture}
  & =
  \kappa \:\begin{tikzpicture}[baseline=10pt,scale=1]
  \qkzplaq(0,0,u_L\!+\!s); \qkzplaq(1,0,u_1); \qkzplaq(2,0,u_2); \qkzplaq(3,0,u_3);\qkzplaq(4,0,\cdots);\qkzplaq(5,0,u_{L-1});
  \wlabels (0,0,a_{L-1},a_L,,); 
  \wlabels (2,0,a_1,a_2,,); 
  \wlabels (5,0,a_{L-2},a_{L-1},,); 
  \end{tikzpicture}
\end{align*}
respectively.
The qKZB equation \eqref{qkzeq} is represented by\\[2mm]
\begin{align*}
\begin{tikzpicture}[baseline=10pt,scale=1.2]
\wsplaq(0,1,u_i\!\!-\!\!u_1\!\!+\!\!s); \wsplaq(1,1,\cdots); \wsplaq(3,1,u_i\!-\!u_{i+1});\wsplaq(4,1,\cdots);\wsplaq(5,1,u_i\!-\!u_L);
\qkzplaq(0,0,u_1); \qkzplaq(1,0,\cdots); \qkzplaq(2,0,u_{i}) \qkzplaq(3,0,u_{i+1});\qkzplaq(4,0,\cdots);\qkzplaq(5,0,u_L);
\wlabels (0,1,a_L,a_1,,); 
\wlabels (2,1,a_{i-1},a_i,,);
\wlabels (3,1,,a_{i+1},,);
\wlabels (5,1,a_{L-1},a_L,,); 
\draw[dashed] (2,1) -- (3,2);
\end{tikzpicture}
&& = \kappa \begin{tikzpicture}[baseline=10pt,scale=1.2]
\qkzplaq(0,0,u_1); \qkzplaq(1,0,\cdots); \qkzplaq(2,0,u_{i}+s); \qkzplaq(3,0,u_{i+1});\qkzplaq(4,0,\cdots);\qkzplaq(5,0,u_L);
\wlabels (0,0,a_{L},a_1,,); 
\wlabels (2,0,a_{i-1},a_i,,);
\wlabels (3,0,,a_{i+1},,); 
\wlabels (5,0,a_{L-1},a_{L},,); 
\end{tikzpicture}.
\end{align*}
Here the dashed line means that heights are identified and we use the convention that all unmarked heights on the left-hand-side are summed over.  

\end{appendix}

\gdef\MRshorten#1 #2MRend{#1}%
\gdef\MRfirsttwo#1#2{\if#1M%
MR\else MR#1#2\fi}
\def\MRfix#1{\MRshorten\MRfirsttwo#1 MRend}
\renewcommand\MR[1]{\relax\ifhmode\unskip\spacefactor3000 \space\fi
\MRhref{\MRfix{#1}}{{\scriptsize \MRfix{#1}}}}
\renewcommand{\MRhref}[2]{%
\href{http://www.ams.org/mathscinet-getitem?mr=#1}{#2}}
\bibliography{sosqkz.bib}
\bibliographystyle{amsplainhyper}

\end{document}